\documentclass[
reprint,
amsmath,amssymb,
aip,
nofootinbib
]{revtex4-2}

\usepackage{amsthm}
\usepackage{graphicx}
\usepackage{float}
\usepackage{dcolumn}
\usepackage{bm}

\usepackage{adjustbox}
\usepackage{ulem}
\usepackage{mathrsfs}
\usepackage{xcolor}
\usepackage{subcaption}
\usepackage{mwe}
\usepackage{autobreak}
\usepackage{multirow}
\usepackage{amsmath}
\usepackage{array}
\newcolumntype{C}[1]{>{\centering\arraybackslash}m{#1}}

\usepackage{CJKutf8}
\newtheorem{thm}{Theorem}

\newtheorem{defn}{Definition}

\usepackage{algorithm,algcompatible,algorithmicx}
\usepackage{algpseudocode}

\algnewcommand\algorithmicinput{\textbf{Input:}}
\algnewcommand\INPUT{\item[\algorithmicinput]}
\algnewcommand\algorithmicoutput{\textbf{Output:}}
\algnewcommand\OUTPUT{\item[\algorithmicoutput]}
\algnewcommand\algorithmicoptional{\textbf{Optional:}}
\algnewcommand\OPTIONAL{\item[\algorithmicoptional]}

\begin{document}


\title{Extending the trapping theorem to provide local stability guarantees for quadratically nonlinear models}%

\author{Mai Peng}
\email{maipeng@uw.edu}
\affiliation{Department of Mechanical Engineering, University of Washington, Seattle, WA, 98195, USA\looseness=-1}

\author{Alan Kaptanoglu}
 \affiliation{Courant Institute of Mathematical Sciences, New York University, New York, NY, 10012, USA \looseness=-1}
\affiliation{Department of Mechanical Engineering, University of Washington, Seattle, WA, 98195, USA\looseness=-1}

\author{Chris Hansen}
\affiliation{ 
Department of Applied Physics and Applied Mathematics, Columbia University, New York, NY, 10027, USA\looseness=-1}%

\author{Jacob Stevens-Haas}
\affiliation{Department of Applied Mathematics, University of Washington,
Seattle, WA, 98195, USA\looseness=-1}

\author{Krithika Manohar}

\author{Steve Brunton}
\affiliation{Department of Mechanical Engineering, University of Washington, Seattle, WA, 98195, USA\looseness=-1}


\begin{abstract}
The Navier Stokes equations (NSEs) are partial differential equations (PDEs) to describe the nonlinear convective motion of fluids and they are computationally expensive to simulate because of their high nonlinearity and variables being fully coupled. Reduced-order models (ROMs) are simpler models for evolving the flows by capturing only the dominant behaviors of a system and can be used to design controllers for high-dimensional systems. However it is challenging to guarantee the stability of these models either globally or locally. Ensuring the stability of ROMs can improve the interpretability of the behavior of the dynamics and help develop effective system control strategies. For quadratically nonlinear systems that represent many fluid flows, the Schlegel and Noack trapping theorem~\cite{Schlegel2015} can be used to check if ROMs are globally stable (long-term bounded). This theorem was subsequently incorporated into system identification techniques that determine models directly from data~\cite{kaptanoglu2021promoting}. While the Schlegel and Noack trapping theorem provides global stability criteria for systems with strictly energy-preserving nonlinearities, many physical systems, including those with inflow/outflow boundary conditions, exhibit weakly relaxed energy-preserving structures. This work introduces two key advances: (1) a theorem establishing analytical stability bounds for linear-quadratic systems under relaxed energy-preserving constraints, explicitly quantifying the local stability radius, and (2) the extended trapping SINDy algorithm, which embeds these theoretical guarantees into data-driven system identification.
By integrating Lyapunov’s direct method with the trapping theorem framework, our approach enables the first provably locally stable models for quadratic dynamics with weakly broken energy-preserving nonlinearities. Several examples are presented to demonstrate the effectiveness and accuracy of the proposed algorithm.
\end{abstract}
\maketitle


\section{\label{sec:Intro}Introduction}
Modeling the full multi-scale spatio-temporal evolution of dynamical systems is often computationally expensive, motivating the use of reduced-order models (ROMs) that capture only the dominant behaviors of a system~\cite{noack2003hierarchy,Noack2011book,carlberg2011efficient,carlberg2013gnat,benner2015survey,Rowley2017arfm}. 
Projection-based model reduction is a common approach for generating ROMs. 
With this technique, a high-dimensional system, such as a spatially discretized set of partial differential equations (PDEs), is projected onto a low-dimensional basis of modes~\cite{Taira2017aiaa,taira2020modal}.  
The projection leads to a computationally efficient system of ordinary differential equations (ODEs) that describes how the mode amplitudes evolve in time~\cite{holmes2012turbulence}. 
However, these models often suffer from stability issues~\cite{Ballarin2015}, causing solutions to diverge in finite-time. In the present work, we refer to models as ``globally stable'' if all trajectories are bounded for all time, and we will formalize this notion in Sec~\ref{sec:ROM}.
The traditional explanation is that the stability issues in Galerkin models derive from truncated dissipative scales, but increasingly there are alternate explanations including fundamental numerical issues with the Galerkin framework, potentially resolved in a Petrov-Galerkin framework, for convection-dominated flows~\cite{grimberg2020stability}.

To determine if a Galerkin model is long-term bounded, Schlegel and Noack~\cite{Schlegel2015} developed a ``trapping theorem'' with necessary and sufficient conditions for long-term model stability for systems that exhibit quadratic, energy-preserving nonlinearities. This theorem can be used as an effective diagnostic for evaluating the stability of reduced-order models. Moreover, the theorem was subsequently incorporated into system identification methods, for producing data-driven models that are provably long-term bounded~\cite{kaptanoglu2021promoting, koike2024energy, liao2024convex,duff2024stability}. The ability to promote global stability guarantees has broad potential for data-driven models, and has already been extended to neural-network-based system identification methods~\cite{Ouala2023}. Recently, Goyal \textit{et al}.\cite{goyal2023guaranteed} extended the trapping theorem in quadratic systems to generalized energy-preserving nonlinearities, achieved through constraints of a special parametrization that guarantee the stability of the learned model by construction. For general quadratic systems, Kramer~\cite{kramer2021stability} discusses the condition for such systems to be locally stable by means of a Lyapunov function and provide better estimates of the radii of the domain of attraction by a novel optimization-based approach. {For systems whose dynamics can be described by more general polynomial functions, the sum-of-squares of polynomials techniques provide a constructive method for generating Lyapunov functions to achieve global stability and control of dynamical systems and ROMs~\cite{chernyshenko2014polynomial, huang2015sum}. However, a sum-of-squares polynomial Lyapunov function for a given system must be searched for and its existence is not guaranteed.}

The trapping theorem assumes that the nonlinear terms do not contribute at all to the dynamical energy evolution. This rather strong assumption may seem somewhat limiting, but actually there are many scenarios under which energy-preserving quadratic nonlinearities can arise. In fluid dynamics, the quadratic nonlinearity represents the convective derivative $(\bm{u}\cdot\nabla)\bm{u}$ in the Navier-Stokes equations (NSEs). This quadratic nonlinearity is energy-preserving for a large number of boundary conditions. Examples include no-slip conditions, periodic boundary conditions~\cite{mccomb1990physics,holmes2012turbulence}, mixed no-slip and periodic boundary conditions~\cite{rummler1998direct}, and open flows in which the velocity magnitude decreases faster than the relevant surface integrals expand (e.g., two-dimensional rigid body wake flows and three-dimensional round
jets)~\cite{schlichting2016boundary}. Inflow and outflow boundary conditions are also common in fluid mechanics, and often only weakly break the energy-preserving structure of the nonlinearities. Unless this constraint is exactly enforced in a numerical scheme, truncation errors and dataset noise will provide additional weak breaking of the nonlinear structure.

In magnetohydrodynamics (MHD), and MHD extensions such as Hall-MHD, there are additional quadratic nonlinearities that are also energy-preserving with common experimental boundary conditions such as a conducting wall~\cite{freidberg2014ideal}, or a balance between dissipation and actuation in a steady-state plasma device~\cite{kaptanoglu2020two,kaptanoglu2021physics}. In this way, the results presented in this work hold for quadratic reduced-order models of convection-dominated fluid flows, Lorentz-force-dominated plasmas, and other quadratic-nonlinearity-dominated dynamical systems. 

\subsection{Contributions of this work}
This work establishes a theorem within the framework of projection-based reduced-order models (ROMs), which guarantees local stability for data-driven models with quadratic nonlinearities. The theorem derives sufficient conditions for the local stability of linear-quadratic systems under weakly relaxed quadratic energy-preserving structures. Furthermore, we demonstrate that the Schlegel and Noack theorem~\cite{Schlegel2015} is recovered as a limiting case by restricting the relaxation of energy-preserving constraints. Formalizing stability under weakly violated energy-preserving nonlinearities is critical for applications involving boundary conditions such as inflow/outflow configurations, which are ubiquitous in fluid systems. Rigorous stability guarantees for this class of systems provide a foundation for designing robust fluid flow control strategies.

We next demonstrate how local stability in data-driven models can be enforced through a modified optimization loss function derived from the theorem. The sparse identification of nonlinear dynamics (SINDy) framework~\cite{brunton2016discover}, known for its utility in parsimonious model discovery, serves as a natural platform for this implementation. We modify the original optimizer based on SINDy for identifying \textit{globally} stable models~\cite{kaptanoglu2021promoting} with minimal adjustments, adapting it to discover \textit{locally} stable systems exhibiting weakly quadratic energy-preserving structures while enabling estimation of the local stability radius. Furthermore, we extend the framework by generalizing quadratic energy-preserving terms through the introduction of a Lyapunov matrix, thereby integrating Lyapunov’s direct method into the Schlegel and Noack trapping theorem.

The rest of the paper is organized as follows: Sec.~\ref{sec:ROM} begins by reviewing energy-preserving quadratic nonlinearities and the Schlegel and Noack trapping theorem, which establishes stability criteria for systems with such structures. In Sec.~\ref{sec:ROM}, we introduce a novel theorem that rigorously quantifies the stability domain for systems where energy-preserving constraints are relaxed—a critical advancement for flows with non-conservative boundary conditions (e.g., open flows). In Sec.~\ref{sec:algorithm}, we develop the extended trapping SINDy algorithm, a significant extension of the original trapping SINDy framework that augments it with local stability guarantees by embedding our theoretical advances into the optimization process. This algorithm is implemented in the open-source PySINDy package~\cite{deSilva2020,Kaptanoglu2022}, enabling reproducible and accessible stable model discovery. Sec.~\ref{sec:results} showcases the algorithm’s ability to recover stable models from noisy data, even under weakened energy-preserving assumptions, across several benchmark systems. Finally, Sec.~\ref{sec:conc} synthesizes the core innovations of this work: (1) the first analytical stability bounds for quadratically nonlinear systems with weakly energy-preserving structure, and (2) a generalizable framework for enforcing stability in data-driven models, bridging Lyapunov theory with the trapping theorem. We conclude with open challenges and future extensions.

\section{\label{sec:ROM}Reduced-order modeling}
In modern scientific computing, a set of governing partial differential equations is typically discretized into a high-dimensional system of coupled ordinary differential equations. 
The state of the system $\bm u(\bm{x},t)\in\mathbb{R}^n$ is defined as a high-dimensional vector that represents the fluid velocity or other set of spatio-temporal fields, for example sampled on a high-resolution spatial grid. For concreteness, in the fluid examples in this work, we will assume that the dynamical flow is evolving according to the incompressible Navier-Stokes equations,  
\begin{subequations}
\label{eq:NSE}
    \begin{align}
    &\nabla \cdot \bm{u} = 0, \\
    &\frac{\partial \bm{u}}{\partial t} = - (\bm{u} \cdot \nabla) \bm{u} + \frac{1}{\rho} \left[  - \nabla p
    + \nu\nabla^2\bm u \right].
\end{align}
\end{subequations}
The density $\rho$ and viscosity $\nu$ are assumed constant, and the pressure is denoted $p$. These equations are expected to accurately approximate incompressible flow of Newtonian fluids with constant density, kinematic viscosity, and inflow/outflow boundary conditions. Other complex fluid flows such as ocean currents, viscoelastic fluids, or ferrofluids are described by models that produce higher order polynomial nonlinearities or trigonometric nonlinear terms, but may be amenable to the quadratic analysis here if regimes exist where higher order nonlinearities are inactive or weak or can be transformed into quadratic nonlinearities~\cite{qian2020lift}.

The goal of a projection-based ROM is to transform this high-dimensional system into a lower-dimensional system of size $r\ll n$ that captures the essential dynamics. 
One way to reduce the set of governing equations to a set of ordinary differential equations is by decomposition into a desired low-dimensional orthogonal basis $\{\bm{\chi}_i(\bm x)\}$ in a process commonly referred to as Galerkin expansion: 
\begin{align}
\label{eq:galerkin_expansion}
    \bm u(\bm{x}, t) = \overline{\bm u}(\bm{x}) + a_i(t) \bm{\chi}_i(\bm{x}).
\end{align}
The usage of Einstein notation is throughout this paper so repeated indices are implicitly summed over. Thus here the sum over $i$ is implied and $i \in \{1, ..., r\}$. Here, $\overline{\bm u}(\bm{x})$ is a base flow which fulfills stationary boundary conditions, $\bm{\chi}_i(\bm x)$ are spatial modes of the expansion for the fluctuation $\bm{u}'(\bm x, t) = \bm{u}(\bm x, t) - \overline{\bm u}(\bm x)$, and $a_i(t)$ describe how the amplitude of these modes vary in time. 
The proper orthogonal decomposition (POD)~\cite{holmes2012turbulence,brunton2019data} is frequently used to obtain the basis, since the modes $\bm{\chi}_i(\bm x)$ are orthogonal and ordered by maximal energy content. 
The governing equations are then Galerkin projected onto the basis $\{\bm{\chi}_i(\bm x)\}$ by substituting Eq.~\eqref{eq:galerkin_expansion} into the PDE of Eq.~\eqref{eq:NSE} and using inner products with $\bm{\chi}_k$ to remove the spatial dependence.  
Orthogonal projection onto POD modes is the simplest and most common procedure, resulting in POD-Galerkin models, although Petrov-Galerkin projection~\cite{carlberg2011efficient,carlberg2013gnat} improves model performance in some cases.

Although POD-Galerkin models can be accurate descriptions of various flows, they are intrusive, in the sense that evaluating the projected dynamics requires evaluating individual terms in the governing equations, such as spatial gradients of the flow fields~\cite{Noack2011book, carlberg2011efficient, carlberg2013gnat, benner2015survey}.
POD-Galerkin models therefore tend to be highly sensitive to factors including mesh resolution, convergence of the POD modes, and treatment of the pressure and viscous terms. For this reason, we pivot to a data-driven approach later in the present work.
%
Lastly, POD-Galerkin models generically suffer from instability issues, which we explore next.


\subsection{Stability definitions}
In this section, necessary background material is introduced to define the problem under discussion in this paper. The goal is to characterize the stability of quadratic dynamical systems with weakly energy-preserving nonlinearities, which involves the use of common analytical methods from basic linear stability analysis to Lyapunov stability. 

We begin with the definition of boundedness. In this paper, we exclusively consider stability in the form of long-term boundedness of dynamical trajectories. Additionally, we will use \textit{globally} long-term bounded to describe dynamical systems for which all trajectories are long-term bounded. \textit{Locally} long-term bounded will refer to a system exhibiting  dynamical trajectories that are bounded if they are initialized within some compact region of $\mathbb{R}^N$.

Consider the first order system
\begin{equation}
\label{eq:general_dynamics}
    \dot{\bm{x}} = \bm f(\bm x)
\end{equation}
where map $f: \mathcal{D} \rightarrow \mathbb R^n$ is locally Lipschitz from {an open and connected set} $\mathcal{D} \subset \mathbb R^n$ onto $\mathbb R^n$. {If such points exist, the equilibrium points or fixed points of the system~\eqref{eq:general_dynamics} are defined as  $\bar{\bm x} \in \mathcal{D}$ such that $f(\bar{\bm x}) = 0$. For convenience, we will state all definitions and theorems for the case when $\bar{\bm x} = 0$, the origin of $\mathbb R^n$. It can be proved that there is no loss of generality in doing so. When discussing the stability of dynamical systems, it is a common practice to use phase space for the intuitive picture of solutions. Now we define a flow on phase space.
\begin{defn}
 \label{def:flow}
The solutions of Eq.~\eqref{eq:general_dynamics} define a flow, $\phi(\bm x, t)$ or $\phi_t(\bm x)$, such that $\phi_t(\bm x)$ is the solution of this differential equation at time $t$ with initial value $\bm x$. Hence,
\begin{equation*}
    \frac{\mathrm{d}}{\mathrm{d} t}\phi(\bm x,t) = f(\phi(\bm x,t)), 
\end{equation*}
for all $t$ such that the solution through $\bm x$ exists and $\phi(\bm x,0) = \bm x$.
\end{defn}
}
{Note that the solution $x(t)$ with $x(0) = x_0$ by this notation is $\phi(x_0,t)$ or $\phi_t(x_0)$. We then give the definitions of terms that describe the stability of dynamical systems that will be discussed in this paper using these notations.}
\begin{defn}
 \label{def:Boundedness}
The solutions of Eq.~\eqref{eq:general_dynamics} are 
\begin{enumerate}
    \item[$\bullet$] {bounded} if there exists a constant $\rho > 0$ that is independent of the initial time $0$, and for every $c \in (0, \rho)$, there is $\mu = \mu(c) > 0$ such that
\begin{equation}
\label{eq:boundedness}
    \|{\bm x}(0) \| \leq c \  \Rightarrow \ \|{\bm x}(t) \| \leq \mu,\ \forall t \ge 0,
\end{equation}
    \item[$\bullet$] {globally bounded} if~\eqref{eq:boundedness} holds for arbitrarily large $\rho$.
    \item[$\bullet$] {If there exists some $\eta$ such that $\eta \geq \mu(c)$ for every $ c \in (0, \rho)$ in the definition above, we can define an upper bound stability region, i.e. a ball center at the origin with radius $\eta$, since \begin{equation}
\label{eq:upperbound}
    \|{\bm x}(0) \| \leq c \  \Rightarrow \ \|{\bm x}(t) \| \leq \eta,\ \forall t \ge 0,
\end{equation} 
meaning that there is an upper bound on how far trajectories can ever be from the origin.
    }
    \item[$\bullet$] 
    {We will also use the term ``domain of boundedness'' to denote the set of all $x(0) = x_0$ such that the solutions $x(t)$ of a given dynamical system are bounded.} 
\end{enumerate}
\end{defn}
%
While local or global boundedness is very general, we will show that for the case of quadratically nonlinear systems, boundedness can be addressed by considering a more specialized class of stability via trapping regions. We will use the Lyapunov function to help define the trapping region, so first we give the definition of the Lyapunov function.

\begin{defn}
 \label{def:Lyapunov_func}
Let $V: \mathbb R^n \rightarrow \mathbb R$ be a continuously differentiable function. Then we can define the derivative of $V(\bm x)$ along trajectories with respect to time~\cite{glendinning1994stability,SWINNERTONDYER2001},
\begin{equation*}
    \dot V = \frac{\mathrm d V}{\mathrm{d}t} = \bm f(\bm x) \cdot \nabla V.
\end{equation*}
Then $V$ is
\begin{enumerate}
    \item[$\bullet$] a Lyapunov function for the value of $c\in \mathbb R_{>0}$ if $\dot V(\bm x) \le 0$ at all points of $V(\bm x) = c$,
    \item[$\bullet$] a strict Lyapunov function for $c$ if $\dot V(\bm x) < 0$ at all points of $V(\bm x) < c$.
\end{enumerate}
\end{defn}

An important case is when there exists some $a,b \in \mathbb R$ such that $V$ is Lyapunov or strict Lyapunov for all $c \in [a,b]$ and if so we denote the set $\{\bm x \in \mathbb R^n \vert V(\bm x) \le a\}$ by $\Omega_{V,a}$. Then each trajectory starting inside $\Omega_{V,b}$ either enters the set $\Omega_{V,a}$ or tends to the set given by $\{\bm x \in \mathbb R^n \vert \dot{V}(\bm x) = 0\}$. We will later see that if $\Omega_{V,a}$ is bounded, it is a trapping region. Now we give the definition of a trapping region.
\begin{defn}
Consider the dynamical system of~\eqref{eq:general_dynamics} and let $\phi_t(\bm x)$ denote a position of a solution of the system at time $t$ that started at $\bm x \in \mathcal{D}$. Then a set $\mathcal{N} \subset \mathcal{D}$ is
\begin{enumerate}
    \item[$\bullet$] a trapping region if it is compact and 
    \begin{equation}
    \label{def:trapping_region}
    \phi_t(\mathcal{N}) \subset \mathrm{int}(\mathcal{N}),\  \forall t>0,
    \end{equation}
    where $\mathrm{int}(\cdot)$ denotes the interior of a set,
    \item[$\bullet$] an attracting trapping region if there exists a compact set $S \subset \mathcal{D}$ and time $T$ such that $\mathcal{N} \subset S$ is a trapping region and
    \begin{equation}
    \phi_t(S\setminus \mathcal{N}) \subset \mathrm{int}(\mathcal{N}),\  \forall t>T,
    \label{def:attracting_trapping_region}
    \end{equation}
    \item[$\bullet$] a monotonically attracting trapping region if $\mathcal{N}$ is an attracting trapping region and there exists a Lyapunov function $V$ such that
    \begin{equation}
    \label{def:mono_attract_TR}
    \dot{V}(\bm x) < 0,\  \forall \bm x\in \mathcal{D} \setminus \mathcal{N}.
    \end{equation}
\end{enumerate}
\end{defn}
An attracting trapping region is one in which trajectories in $\mathcal{D} \setminus \mathcal{N}$ eventually fall into $\mathcal{N}$, while a monotonically attracting trapping region is one in which trajectories in $\mathcal{D} \setminus \mathcal{N}$ are \textit{always} falling towards $\mathcal{N}$. Also, a trapping region is not exactly an attractor, as a trapping region is a neighborhood of an attractor $\mathcal A$ that satisfies $\phi_t(\mathcal{A}) \subset \mathcal{A}$ for all $t$ and no proper subset of $\mathcal A$ is so. Fig.~\ref{fig:trapping_sketch} illustrates this difference. Consider $V$ a Lyapunov function and each dashed closed surface $V = V(\bm x) = c_i\ (i=1,2,3,4)$ is a Lyapunov surface with $c_1 > c_2 > c_3 > c_4$. For each Lyapunov surface outside $\bm{\mathcal{N}}$, $\bm \Gamma_2$ only passes through at most once, while $\bm \Gamma_1$ crosses $V = c_3$ more than once before each enters into $\bm {\mathcal{N}}$. If we consider the Lyapunov function as a form of energy, along $\bm \Gamma_2$, the energy is always decreasing before entering $\bm{\mathcal{N}}$. 
    If all trajectories starting outside $\bm{\mathcal{N}}$ behave the same as $\bm \Gamma_2$, then $\bm{\mathcal{N}}$ is a monotonically attracting trapping region, otherwise an attracting trapping region. These definitions form the basis for our analysis following in the next a few sections.
\begin{figure}
    \centering
    \includegraphics[width=0.8\linewidth]{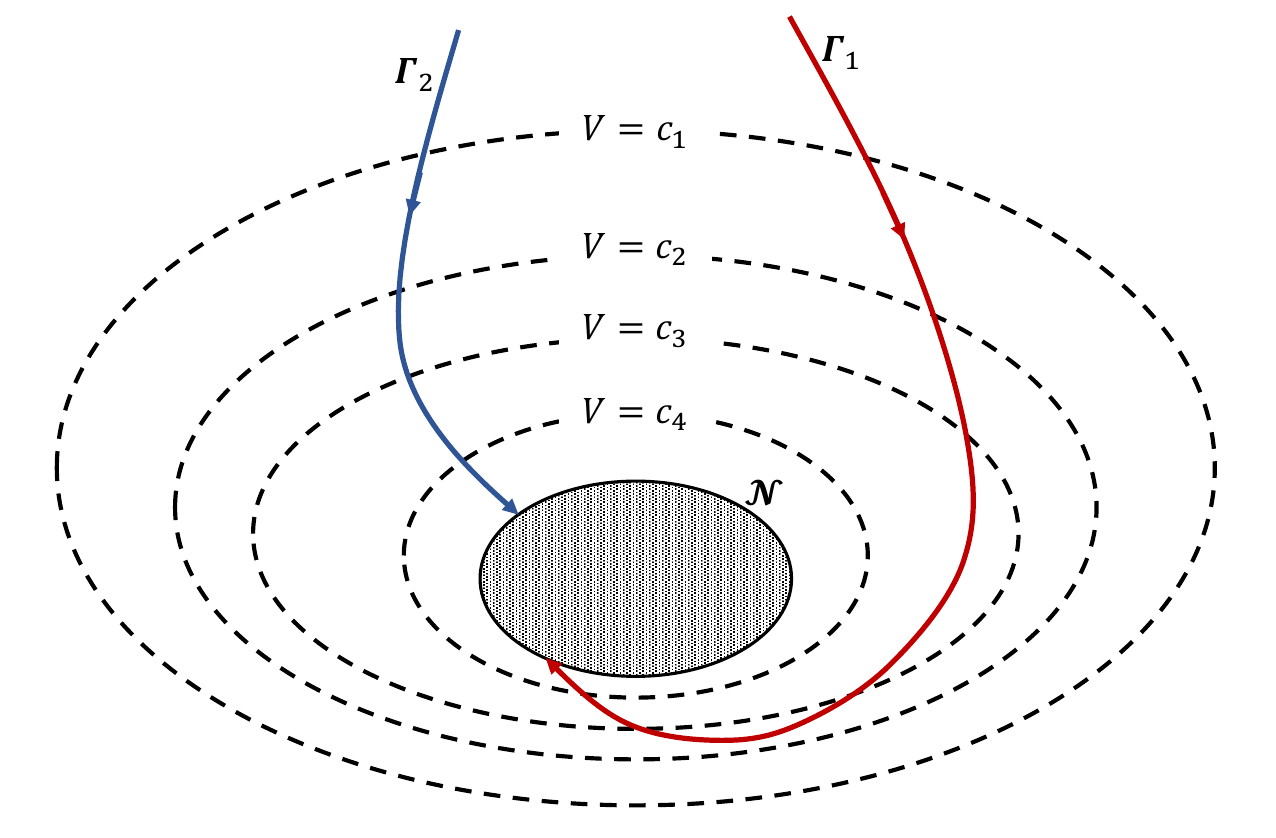}
    \caption{Sketch of an attracting trapping region. $\bm \Gamma_1$ (in red, solid) and $\bm \Gamma_2$ (in blue, solid) are two trajectories starting outside an attracting trapping region $\bm{\mathcal{N}}$ where all trajectories will eventually enter in and stay there forever.}
    \label{fig:trapping_sketch}
\end{figure}
\subsection{\label{sec:EPN}Flows with energy-preserving nonlinearities}

Now, if the governing system of partial differential equations for $\bm u(\bm{x},t)$ are at most quadratic in nonlinearity, as for many fluids, plasmas, and other dynamical systems, Galerkin projection onto the governing equations produces the following system of ODEs for the set of temporal functions $ a_i(t)$,
\begin{align}
\label{eq:Galerkin_model}
\dot{a}_i(t) &= E_i+ L_{ij}a_j + Q_{ijk}a_ja_k.
\end{align}
$E_i$, $L_{ij}$, and $Q_{ijk}$ are tensors of time-independent coefficients, obtained from spatial inner products between the $\bm{\chi}_i(\bm x)$ and the spatiotemporal operators that define the model dynamics. We have assumed that the pressure gradient term in Eq.~\eqref{eq:NSE} has been removed as in Rowley \textit{et al}.~\cite{Rowley2004pd}. Moreover, Eq.~\eqref{eq:Galerkin_model} can be proved that it is locally Lipschitz since it is continuously differentiable over any compact domain in a Euclidean space; therefore the existence and uniqueness of the initial value problem are guaranteed. The class of systems considered for the original trapping theorem are those with energy-preserving nonlinearity, which is defined as
\begin{align}
\label{eq:energy_preserving_nonlinearity_full}
    Q_{ijk}a_ia_ja_k = 0, 
\end{align}
or equivalently, for all $i,j,k \in \{1,...,r\}$,
\begin{align}
\label{eq:energy_preserving_nonlinearity}
Q_{ijk} + Q_{jik} + Q_{kji}= 0.
\end{align} 
The $Q_{ijk}$ are symmetric in swapping $j$ and $k$ without loss of generality. To see why this condition preserves the energy, consider $\overline{\bm u}(\bm x)$ being the mean flow and we can introduce the standard kinetic energy $K$,
\begin{align}
    K = \frac{1}{2}\langle \bm u, \bm u \rangle .
\end{align}
Then the change of the kinetic energy in time is,
\begin{align}
\label{eq:energy_deriv_pod}
    \dot{K} &= \langle \bm u, \dot{\bm u} \rangle = \langle \overline{\bm u}(\bm{x}) +  a_i(t) \bm{\chi}_i(\bm{x}), \dot{a}_j(t) \bm{\chi}_j(\bm{x}) \rangle \\ \notag
    &= (a_i - m_i)\dot{a}_i = (a_i - m_i)\left(E_i+ L_{ij}a_j + Q_{ijk}a_ja_k\right), \\
    \label{eq:m_fixed}
    m_j &\equiv - \langle \overline{\bm u}(\bm x), \bm \chi_j(\bm x) \rangle.
\end{align}
The only cubic term in the temporal modes in Eq.~\eqref{eq:energy_deriv_pod} is precisely $Q_{ijk}a_ia_ja_k$ and by Eq.~\eqref{eq:energy_preserving_nonlinearity_full} it vanishes. The eradication of any cubic terms in $\bm a$, as well as the shift of the energy via $\bm m$ from the presence of a steady-state contribution $\overline{\bm u}(\bm{x})$, are critical components for the definition of a suitable Lyapunov function, as is illustrated below.
Note that this analysis is unchanged if the modes $\bm \chi_j$ are not orthogonal, since this only introduces a mass matrix~\cite{rempfer1994dynamics} $M_{ij} = \langle \bm \chi_i, \bm \chi_j \rangle$, which is symmetric by construction and leaves the energy evolution qualitatively unchanged if it is positive definite~\cite{kramer2021stability}.

\subsection{\label{sec:trapping}Schlegel and Noack trapping theorem}

The Schlegel and Noack trapping theorem provides necessary and sufficient conditions for energy-preserving, effectively nonlinear, quadratic systems {to exhibit a monotonically trapping region, $\bm B(\bm{m},R_m)$, a ball centered at $\bm{m} \in \mathbb{R}^r$ with radius $R_m$. } 
Outside this region the rate of change of energy $K$ is negative everywhere, producing a Lyapunov function that renders this system globally stable.  
The trapping theorem begins by recentering the origin by a now arbitrary constant vector $\bm{m}$ so that we can define a shifted energy  expressed in terms of the shifted state vector $\bm{y}(t)=\bm{a}(t)-\bm{m}$,
\begin{align}
\label{eq:K}
    K = \frac{1}{2}y_i^2.
\end{align}
Taking a derivative and substituting in Eq.~\eqref{eq:Galerkin_model} produces Eq.~\eqref{eq:energy_deriv_pod}, which can be written as
\begin{subequations}
    \begin{align}
\label{eq:Kdot}
\dot{K} &= A_{ij}^Sy_iy_j + d_iy_i, \\
\label{eq:def_AS_LS_dm}
    A_{ij}^S &= L^S_{ij} + \left(Q_{ijk} + Q_{jik} \right)m_k, \\  L_{ij}^S &= \frac{1}{2}(L_{ij} + L_{ji}),\\
\label{eq:bias_d}
    d_i &= E_i + L_{ij}m_j + Q_{ijk}m_im_j.
\end{align}
\end{subequations}

Exploiting that $Q_{ijk}$ is symmetric in the last two indices and energy-preserving, $\bm{A}^S$ can also be expressed as,
\begin{subequations}
\label{eq:preservingAs}
 \begin{eqnarray}
     A_{ij}^S &=& L^S_{ij} - Q_{kij}m_k, \ \  \text{or}\\
     \bm{A}^S &=& \bm{L}^S - \bm m^T\bm Q.
 \end{eqnarray}
\end{subequations}
Then the trapping theorem may now be stated as:
\begin{thm}
\label{th:trapping_theorem}
Consider an effectively quadratically nonlinear system with energy-preserving
nonlinearity. There exists a monotonically trapping ball $\bm B(\bm{m},R_m)$ if and only if the real, symmetric matrix $\bm{A}^S$ is negative definite (\textit{Hurwitz}) with eigenvalues $\lambda_r \leq \cdots \leq \lambda_1 < 0$; the radius is then given by $R_m = \|\bm{d}\|_2/|\lambda_1|$.\footnote{If a system is long-term bounded (not necessarily exhibiting a monotonically trapping region) and effectively nonlinear,  the existence of an $\bm{m}$ producing a negative \emph{semi}definite $\bm{A}^S$ is still guaranteed.} 
\end{thm}

In practice, the goal is then to find an origin $\bm{m}$ so that the matrix $\bm{A}^S$ is \textit{Hurwitz}, guaranteeing a trapping region and global stability. {Moreover, we may want a tight bound on the size of this trapping ball, which can be achieved by varying $\bm m$ to reduce $R_m$. If $R_m \gg 1$, the existence of a trapping ball tells us relatively little about the long-time dynamical trajectories except for the guarantee of boundedness.} There is also a caveat for systems that do not exhibit effective nonlinearity; effective nonlinearity means that there is no linear subspace in which a trajectory starts and remains indefinitely, so that the quadratic nonlinearities $Q_{ijk}y_jy_k$ are never activated. 

%
Without effectively quadratic terms, only the sufficiency in Thm~\ref{th:trapping_theorem} holds; if we can find an $\bm{m}$ so that $\bm{A}^S$ is \textit{Hurwitz}, the system exhibits a trapping region. However, such systems can be globally stable without admitting such an $\bm{m}$. 
Even when the requirements of the trapping theorem are not fully satisfied, system identification work that incorporates this stability theorem~\cite{kaptanoglu2021promoting} still appears to produce models with improved stability properties. The same reference contains an extended discussion of the circumstances under which this theorem holds. 

\subsection{\label{sec:stability_analysis}Analytical estimates of the local stability of linear-quadratic systems}

In this section, we use Lyapunov's direct method to explore the stability of a linear-quadratic system of the form of Eq.~\eqref{eq:Galerkin_model}. We consider the situation where the nonlinearities are no longer exactly energy-preserving for a linear-quadratic system. However, we assume that the breaking of the energy-preserving nonlinearity is bounded in Frobenius norm by a small amount $\epsilon_Q$, i.e.,
\begin{align}
    \| Q_{ijk} + Q_{jik} + Q_{kji} \|_F \leq \epsilon_Q,
\end{align}
which means that it slightly breaks the energy-preserving property of the quadratic term of the dynamical system. Hence, the cubic term in $\dot K$ does not vanish in this case. Based on this constraint, consider the linear-quadratic system,
\begin{align}
    \label{eq:lqsystem_Q}
    \dot{\bm a} &= \bm E + \bm L\bm a + \bm Q\bm{aa},
\end{align}
with state variable $\bm a = \bm a(t) \in \mathbb R^r$, bias term $\bm E \in \mathbb R^r$, matrix $\bm L \in \mathbb R^{r\times r}$ and third-order tensor $\bm Q \in \mathbb R^{r\times r\times r}$. For generalization of Lyapunov's direct method, we introduce an arbitrary state $\bm m$ as the new shifted origin, so that the dynamical system expressed in this new translated coordinate $\bm y = \bm a - \bm m$ is,
\begin{align}
    \label{eq:c_L_Q^System}
    \dot{\bm y} = \bm d + \bm A \bm y + \bm Q\bm y \bm y,
\end{align}
with $\bm d$ defined in Eq.~\eqref{eq:bias_d} and
\begin{align}
    A_{ij} = L_{ij} + (Q_{ijk} + Q_{ikj})m_k.
\end{align}

For the convenience of applying Lyapunov stability analysis, we now introduce a transformation
\begin{equation}
    \label{eq:transform_T}
    \mathcal{T}: \mathbb{R}^{r \times r \times r} \longrightarrow \mathbb{R}^{r \times r^2},
\end{equation}
where the operator $\mathcal{T}$ vectorizes the last two indices of the third order tensor with dimension $r\times r\times r$ by stacking the columns of the corresponding matrix to get a new matrix with the dimension $r\times r^2$. Assume $\bm H = \mathcal{T}(\bm Q)$, then the system of Eq.~\eqref{eq:lqsystem_Q} can be written as,
\begin{align}
\label{eq:shifted_dynamical^System}
    \dot{\bm y} &= \bm d + \bm A\bm y + \bm H(\bm{y\otimes y}),
\end{align}
where $\otimes$ denotes the standard Kronecker product. Notice that we previously assume that the $Q_{ijk}$ are symmetric in the last two indices, hence $\bm H$ is symmetric, i.e., $\bm H(\bm y_1 \otimes \bm y_2) = \bm H(\bm y_2 \otimes \bm y_1)$ for arbitrary $\bm y_1, \bm y_2 \in \mathbb R^r$, which does not change the dynamics~\cite{Two-SidedProjectionMethods}.
To generalize the discussion beyond the shifted energy, we define a Lyapunov function $v(\bm y)$ given a positive definite matrix $\bm P$,
\begin{align}
    \label{eq:lyapunov_function}
    v(\bm y) &= \frac{1}{2} \|\bm y\|^2_{\bm P} =  \frac{1}{2} \langle \bm y, \bm y \rangle_{\bm P}
    = \frac{1}{2}\bm y^T\bm P \bm y,
\end{align}
with the time derivative of $v(\bm y)$
\begin{align}
\label{eq:vdot}
    \dot v(\bm y)
    &= \bm y^T\bm P  \dot{\bm y}.
\end{align}
{This matrix $P$ is any positive definite matrix that forms the equation~\cite{khalil2002nonlinear}
\begin{align*}
    \label{eq:Lyapunov_eq}
    \bm A^T\bm P + \bm P \bm A + \bm M = 0,
\end{align*}
where $A$ is \textit{Hurwitz} and $\bm M$ is positive definite. We call such positive definite matrix $\bm P$ the Lyapunov matrix.} Importantly, the traditional trapping theorem uses the Lyapunov matrix $\bm P = \bm I$, the identity, but for two-dimensional fluid flows the enstrophy is an invariant with $\bm P \neq \bm I$ and it can also be immediately used in the following results. Similar invariants exist in MHD, and more general Lyapunov matrices $\bm P$ may exist for certain flows. Section~\ref{sec:von_karman_enstrophy} will illustrate provably locally-stable, data-driven models of the von Karman vortex street using the enstrophy version of the various trapping theorems described in the present work.

Since the domain of boundedness $\bm \Omega_{\bm m}$ can be complicated geometrically, we consider closed norm balls $\textbf B \subseteq \bm \Omega_{\bm m}$ as the candidates of regions where the solutions of Eq.~\eqref{eq:shifted_dynamical^System} are bounded,{
\begin{align*}
    \textbf B_{\bm y}(R) = \{\bm y \in \mathbb R^r\ |\ \|\bm y\|_{\bm P}\le R\}, \quad \|\bm y\|_{\bm P} = \sqrt{\bm y^T\bm P \bm y}.
\end{align*}}
{Note that
\begin{equation}
\label{eq:shifted_lyapunov}
    v(\bm y) = \frac{1}{2} \bm{y}^T\bm{P}\bm y = \frac{1}{2}\bm z^T\bm z,\quad \bm z=\bm P^{\frac{1}{2}} \bm y
\end{equation}
where $\bm P^{\frac{1}{2}}$ and its inverse exist (and are also symmetric positive-definite) because $\bm P$ is symmetric positive-definite. It follows that
\begin{align}
\label{eq:shifted_ball}
    \textbf B_{\bm y}( R) = \textbf B_{\bm z}( R) = \left\{\bm z \in \mathbb R^r\ \vert\ \|\bm z\|_{2}\le  R \right\}.
\end{align}}
We now provide a lower-bound for the domain of boundedness by the set $\textbf B_{\bm y}$. The task then becomes to make this bound as tight as possible, which is, to find $\sup{ R}$ as the estimate of the local stability radius. This leads to the following theorem.
\begin{thm}
\label{thm:trapping_radius}
    Regarding the system \eqref{eq:c_L_Q^System} and a Lyapunov matrix $\bm P$, there exists a trapping region if the following statements are true:
    \begin{itemize}
        \item[a.] 
        {$\bm{\tilde{A}}^S = \frac{1}{2}\left(\bm P^{\frac{1}{2}}\bm A \bm P^{-\frac{1}{2}}+\bm P^{-\frac{1}{2}} \bm A^T \bm P^{\frac{1}{2}}\right) $ is Hurwitz, i.e. $\lambda_r(\bm{\tilde{A}}^S) \le ... \le \lambda_1(\bm{\tilde{A}}^S)< 0$.
        \item[b.] $\lambda^2_1 - \frac{4\epsilon_Q}{3}\|\bm{\tilde{d}}\|_2 \ge 0$, where
            $\epsilon_Q = \|\bm {\tilde{H}}_0\|_F$, $\bm {\tilde{H}}_0 = \bm P^{-\frac{1}{2}}\mathcal{T}(\tilde Q_{ijk} + \tilde Q_{ijk} + \tilde Q_{kji}) (\bm P^{-\frac{1}{2}}\otimes \bm P^{-\frac{1}{2}})$, $\bm {\tilde Q} = \bm{PQ}$ and $\tilde{\bm d} = \bm P^{\frac{1}{2}}\bm d$. }
    \end{itemize}
    Then the smallest radius of the trapping region $\bm{\mathrm{B}}_{\bm y}(\rho_-)$ is given by
    \begin{align}\label{eq:rho_minus}
    \rho_- = \frac{3|\lambda_1|}{2\epsilon_Q} \left( 1 - \sqrt{1 - \frac{4\epsilon_Q}{3\lambda_1^2}\|\bm{\tilde{d}}\|_2}  \right),
    \end{align}
    while all trajectories starting inside $\bm{\mathrm{B}}_{\bm y}(\rho_+)$ are bounded where \begin{align}\label{eq:rho_plus}
    \rho_+ = \frac{3|\lambda_1|}{2\epsilon_Q} \left( 1 + \sqrt{1 - \frac{4\epsilon_Q}{3\lambda_1^2}\|\bm{\tilde{d}}\|_2} \right).\end{align}
\end{thm}
\begin{proof}
\label{prf:trapping_radius}
The time derivative of the Lyapunov function given by Eq.~\eqref{eq:vdot} is {
    \begin{align*}
    \dot v(\bm y) &= \bm y^T\bm P  \dot{\bm y}\\ 
    &= \bm{d}^T\bm P\bm y + \frac{1}{2}\bm y^T(\bm P\bm {A} + \bm A^T \bm P)\bm y + \bm y^T\bm P\bm H(y\otimes y)
\end{align*}
We can then introduce $\tilde Q'_{ijk} = \tilde Q_{ijk} + \tilde Q_{jki} + \tilde Q_{kij}$, and $\bm {H}_0 = \mathcal{T}(\tilde {\bm Q}')$. See that 
\begin{align*}
    \bm y \bm{H}_0(\bm y\otimes \bm y) &= \sum_{i,j,k}\tilde Q'_{ijk}y_iy_jy_k \\
    &= \sum_{i,j,k}(\tilde Q_{ijk} + \tilde Q_{jki} + \tilde Q_{kij})y_iy_jy_k \\
    &= 3\sum_{i,j,k}\tilde Q_{ijk}y_iy_jy_k \\
    &= 3\bm y^T\bm P\bm H(y\otimes y)
\end{align*}
Then by Eq.~\eqref{eq:shifted_lyapunov}, $\dot v(\bm y)$ can also be written as
\begin{subequations}
    \begin{align*}
        \dot v(\bm z) &= \bm{\tilde{d}}^T\bm z + \bm z\bm{\tilde{A}}^S \bm z + \frac{1}{3}\bm z^T\bm P^{-\frac{1}{2}}\bm H_0 (\bm P^{-\frac{1}{2}}\otimes \bm P^{-\frac{1}{2}})(\bm z \otimes\bm z)\\
        & = \bm{\tilde{d}}^T\bm z + \bm z\bm{\tilde{A}}^S \bm z + \frac{1}{3}\bm z^T\bm {\tilde{H}}_0(\bm z \otimes\bm z)
    \end{align*}
\end{subequations}
Here $\bm{\tilde{d}} = \bm P^{\frac{1}{2}}\bm d$, $\bm{\tilde{A}}^S = \frac{1}{2}\left(\bm P^{\frac{1}{2}}\bm A \bm P^{-\frac{1}{2}}+\bm P^{-\frac{1}{2}} \bm A^T \bm P^{\frac{1}{2}}\right)$ and $\bm {\tilde H}_0 = \bm P^{-\frac{1}{2}}\bm H_0 (\bm P^{-\frac{1}{2}}\otimes \bm P^{-\frac{1}{2}})$. We next consider the region where its derivative is negative. It follows that
\begin{subequations}
    \begin{align}
    \label{eq:shifted_energy}
    \dot v(\bm z) &= \bm{\tilde{d}}^T\bm z + \bm z\bm{\tilde{A}}^S \bm z + \frac{1}{3}\bm z^T\bm {\tilde{H}}_0(\bm z \otimes\bm z) \\ \label{eq:second_Step}
    &\leq \|\bm{\tilde{d}}^T\bm z\|_{2} + \lambda_1 \|\bm z\|_{2}^2 + \frac{1}{3}\|\bm {\tilde{H}}_0\|_2\|\bm z\|_{2}^3\\
    \label{eq:final_quadratic_equation}
    &\leq \|\bm{\tilde{d}}\|_2\|\bm z\|_{2} + \lambda_1 \|\bm z\|_{2}^2 + \frac{\epsilon_Q}{3}\|\bm z\|_{2}^3.
\end{align}
\end{subequations}
In the second step, Eq.~\eqref{eq:second_Step}, we replaced each term on the right-hand side with an upper bound on that term, and used the definition of a matrix 2-norm. In the last step, Eq.~\eqref{eq:final_quadratic_equation}, we used the Cauchy-Schwartz inequality and the matrix norm inequality $\|\bm {\tilde{H}}_0\|_2 \leq \|\bm {\tilde{H}}_0\|_F$, where $\|\cdot \|_F$ denotes the Frobenius norm of the matrix.
For our upper bound estimate on $\dot{v}(\bm z)$ to be negative, we consider solving $\|\bm{\tilde{d}}\|_2\|\bm z\|_{2} + \lambda_1 \|\bm z\|_{2}^2 + \frac{\epsilon_Q}{3}\|\bm z\|_{2}^3 = 0$, from which it follows that
\[
    \rho_{-} < \|\bm z\|_2 < \rho_{+},
\]
\begin{align}
\notag
    \rho_{+} &= \frac{3}{2\epsilon_Q} \left( \sqrt{\lambda^2_1 - \frac{4\epsilon_Q}{3}\|\bm{\tilde{d}}\|_2} - \lambda_1 \right),
    \\ \notag
    \rho_{-} &= -\frac{3}{2\epsilon_Q} \left( \sqrt{\lambda^2_1 - \frac{4\epsilon_Q}{3}\|\bm{\tilde{d}}\|_2} + \lambda_1 \right).
\end{align}
These equations can be manipulated into Equations~\eqref{eq:rho_minus} and~\eqref{eq:rho_plus}.
To guarantee $\rho_+$ and $\rho_-$ both being real, $\Delta = \lambda^2_1 - \frac{4\epsilon_Q}{3}\|\bm{\tilde{d}}\|_2$ must be non-negative. For all $\|\bm y\|_{\bm P} = \|\bm z\|_2 \in (\rho_-, \rho_+)$, we now have $v(\bm y) > 0$ and $\dot v(\bm y) < 0$, which yields the claimed result.}
\end{proof}

\subsection{\label{subsec:interp}Interpretation of the local stability radius}
The $\rho_+$ and $\rho_-$ are two nontrivial roots of Eq.~\eqref{eq:final_quadratic_equation}. As shown in the Fig.~\ref{fig:stability_sketch}, $\rho_+$ is a conservative estimate of the diameter of the inscribed ball $\textbf B_{\rho_+}$ in the locally stable region $\bm{\Omega}_{\rho_+}$. $\partial \bm \Omega_{\rho_+}$ has a complex geometry but can be bounded below by $\partial \bm B_{\rho_+}$. Meanwhile $\rho_-$ is an estimate of the diameter of the set $\Omega_{\rho_-}$ and $\dot v = 0$ at the boundary $\partial \Omega_{\rho_-}$. $\bm \Omega_{\rho_-}$ has $\dot{v} > 0$, a shape close to ellipsoid, and can be bounded above by $\partial \bm B_{\rho_-}$. Negative energy growth in $\textbf B_{\rho_+}\setminus \textbf B_{\rho_-}$ {which is an elliptical annulus in the state space of $\bm y$}, is guaranteed according to Thm.~\ref{thm:trapping_radius} and inside $\textbf B_{\rho_-}$ there exists a region $\bm \Omega_{\rho_-}$, inside of which positive energy growth can be observed. Then the long-time behavior of solutions to the system~\eqref{eq:shifted_dynamical^System} with initial conditions $\bm P^\frac{1}{2}\bm y_0 = \bm z_0 \in \textbf B_{\rho_+}$ can alternate between $\textbf B_{\rho_-}\setminus \bm \Omega_{\rho_-}$ and $\bm \Omega_{\rho_-}$, as shown in the Fig.~\ref{fig:trajectory_sketch}. In a special case where $\bm z_0$ is a real solution to Eq.~\eqref{eq:shifted_energy}, the long-time behavior of the system stays at the boundary of $\bm \Omega_{\rho_-}$.
It is not difficult to conclude that
\begin{align}
\label{eq:limit}
    \lim_{\epsilon_Q \to 0} \rho_{+} &= \infty, \quad
    \lim_{\epsilon_Q \to 0} \rho_{-} = \frac{\|\tilde{\bm{d}}\|_2}{|\lambda_1(\tilde{\bm{A}}^S)|},
\end{align}
which correctly reproduces the original trapping theorem guaranteeing global boundedness (in fact, a monotonically trapping region). We can also validate the limit by setting $\epsilon_Q$ to zero so that $\bm  \tilde {\bm H} = 0$, then the time derivative of the Lyapunov function \eqref{eq:shifted_energy} becomes
\begin{align*}
    \label{eq:noepsilonvdot}
    \dot v(\bm z) = \bm{\tilde{d}}^T\bm z + \bm z^T\bm{\tilde{A}}^S \bm z
    &\le \|\bm{\tilde{d}}\|_2\|\bm z\|_{2} + \lambda_1 \|\bm z\|_{2}^2\\
    &= \|\bm z\|_2\left (\|\bm{\tilde{d}}\|_2 + \lambda_1\|\bm z\|_2\right),
\end{align*}
from which we can see that for all $\|\bm z\|_2 > \|\tilde{\bm d}\|_2 /|\lambda_1|$, $\dot v < 0$. And an estimate of the radius of the ball-shape trapping region is given by $R_{\bm m} = \|\tilde{\bm d}\|_2 /|\lambda_1|$.
\begin{figure}
    \centering
    \includegraphics[width=1\linewidth]{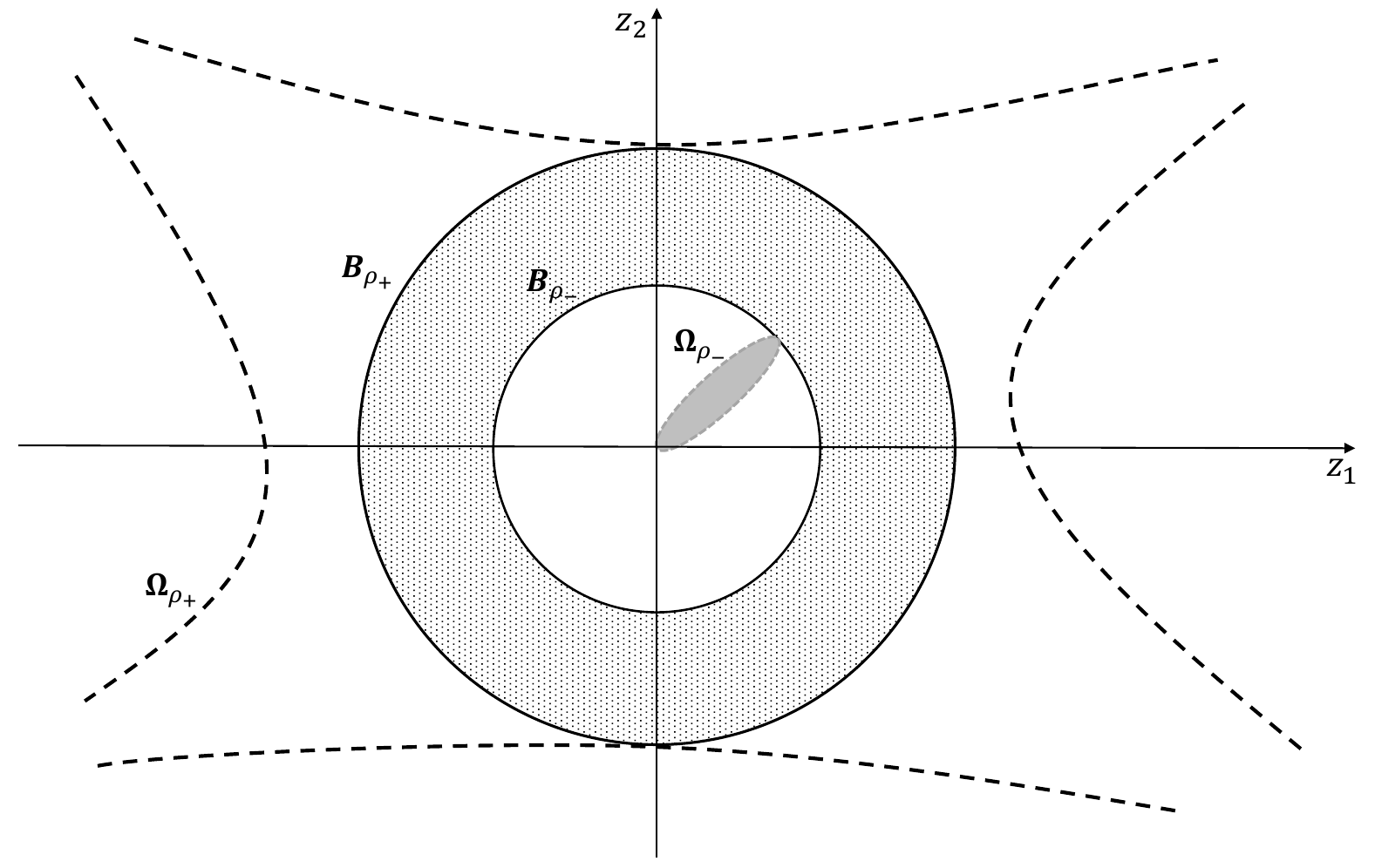}
    \caption{Sketch of the trapping region. $\dot v < 0$ is observed in $\bm \Omega_{\rho_+}\setminus\bm \Omega_{\rho_-}$ and long-term stability is guaranteed for any initial conditions $ \|\bm y(t_0)\|_{\bm P} =\|\bm z(t_0)\|_2 \leq \rho_+$.}
    \label{fig:stability_sketch}
\end{figure}
 \begin{figure}
    \centering
    \includegraphics[width=\linewidth]{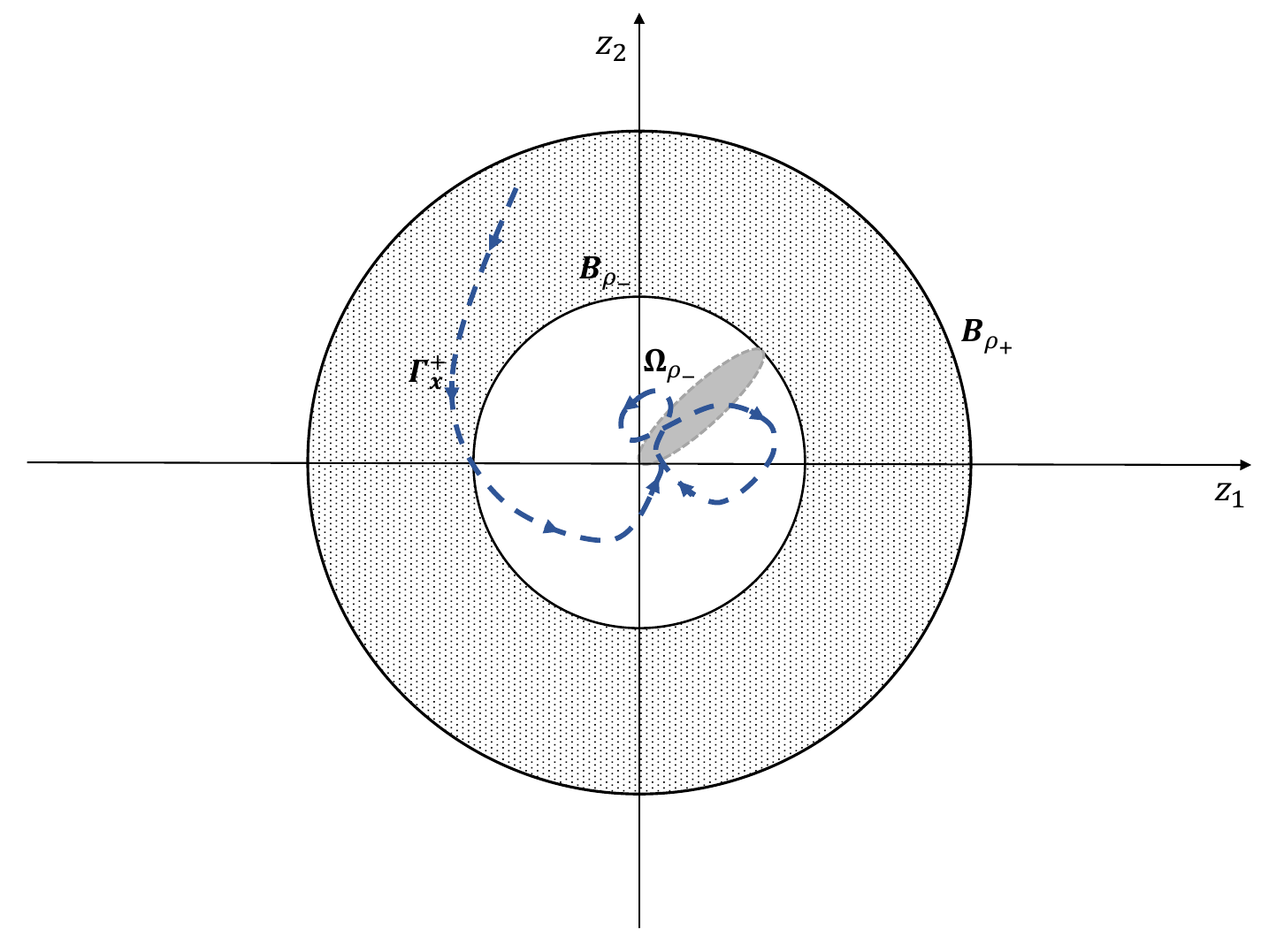}
    \caption{An example trajectory (dashed blue) beginning outside of the locally attracting trapping region $\bm B_{\rho_-}$ and remaining in once it enters in it.}
    \label{fig:trajectory_sketch}
\end{figure} 
Note that the ``shell'' of stability $\textbf B_{\rho_+}\setminus\textbf B_{\rho_-}$ vanishes when $\rho_+ = \rho_-$ or equivalently when the square root vanishes:
\begin{align}
\epsilon_Q = 
    \frac{3\lambda^2_1(\tilde{\bm{A}}^S)}{4\|\tilde{\bm{d}}\|_2},
\end{align}
which provides a condition on the smallness of the totally symmetric part of the quadratic nonlinearity in order for this shell of stability to exist. We use this condition in Appendix~\ref{sec:sensitivity_analysis} to evaluate the sensitivity of this stability estimate to small perturbations in the model coefficients and illustrate evidence that the problem of retaining stability in machine-learned quadratically nonlinear models is generally ill-conditioned. 
%
Note when $\|\bm d\|_2 = 0$, $\bm \Omega_{\rho_-}$ collapses to an asymptotically stable fixed point $\bm m$. 

\section{\label{sec:algorithm}Extended trapping SIND\MakeLowercase y algorithm}
In this section, we develop the extended trapping SINDy algorithm building on Thm.~\ref{thm:trapping_radius}, which enforces local stability guarantees during data-driven model identification. Unlike intrusive methods like POD-Galerkin, which require costly field-derivative evaluations, our approach directly constructs stable models from data via sparse optimization~\cite{brunton2019data, holmes2012turbulence}. 

In particular, we use the trapping SINDy algorithm~\cite{kaptanoglu2021promoting} and formulate a modified ``extended'' trapping SINDy algorithm. We first review the standard SINDy framework~\cite{brunton2016discover} and the trapping SINDy algorithm, which successfully incorporates the trapping theorem (Thm.~\ref{th:trapping_theorem}) into model identification in Sec.~\ref{subsec:std_algrthm}. We then introduce our proposed extended trapping SINDy algorithm in Sec~\ref{subsec:extended_trapping}.

\subsection{\label{subsec:std_algrthm}Standard and trapping SINDy algorithm}

Identifying dynamical systems from data has been an ongoing challenge in mathematical physics, especially in the study of fluid dynamics. The SINDy algorithm leverages the fact that for most physical systems, only a few relevant terms determine the dynamics, to identify parsimonious models from data while avoiding over-fitting. To determine the system of Eq.~\eqref{eq:Galerkin_model} from data, we build SINDy models for the dynamics of $\bm a$. And like most physical systems, we assume the governing equations of dynamics of $\bm a$ consist of only a few terms from a library $\bm \Theta$ consisting of candidate nonlinear functions of the state variable, which makes it sparse in the nonlinear function space. For the dynamics considered in the present work, the system of ODEs for $\bm a$ is at most quadratic in nonlinearity, so that the dynamics will be described as,
\begin{equation}
\label{eq:num_dynamics}
    \frac{d}{dt}\bm a = \bm \Theta(\bm a)\bm \Xi,\ \ \ \ \bm \Theta(\bm a) = \begin{bmatrix}
 | & | & |\\
 1 & \bm a & \bm{a}^{P_2}\\
 | & | & |
\end{bmatrix}.  
\end{equation}
Here, $\bm{a}^{P_2}$ denotes the quadratic polynomial in the state $\bm a$. Consider a collection of snapshots of state $\bm a$ and $\dot{\bm a}$ which come from either measurement or numerical approximation from $\bm a$,
\begin{align*}
    \bm X &= \begin{bmatrix}
        \bm a(t_1) & \bm a(t_2) & \cdots & \bm a(t_M)
    \end{bmatrix}^T, \\
    \dot{\bm X} &= \begin{bmatrix}
        \dot{\bm a}(t_1) & \dot{\bm a}(t_2) & \cdots & \dot{\bm a}(t_M)
    \end{bmatrix}^T.
\end{align*}
The dynamics of Eq.~\eqref{eq:num_dynamics} become $\dot{\bm X} =\bm \Theta (\bm X)\bm \Xi$, where $\bm \Xi = [\bm \xi_1\ \bm \xi_2\ \cdots\ \bm \xi_r]$ is a sparse matrix with each $\bm \xi_j \in \mathbb R^N$ indicating which of the corresponding candidate functions in the library $\bm \Theta \in \mathbb R^{M\times N}$ are active in the $\dot{a}_j$ equation. One can find that the quadratic library has $N = \frac{1}{2}(r^2 + 3r) + 1$ terms. Now we have the sparse optimization problem introduced in the standard SINDy algorithm:
\begin{equation}
    \bm \xi_j = \operatorname*{argmin}_{\bm \xi_j'}\left [\|\dot{\bm X}_{:,j}-\bm \Theta(\bm X)\bm \xi_j'\|_2 + \gamma\|\bm \xi_j'\|_0\right ],
\end{equation}
where $\dot{\bm X}_{:,j}$ denotes the $j^{th}$ column of $\dot{\bm X}$, $\|\cdot\|_0$ counts the number of nonzero elements and $\gamma$ determines the strength of sparsity-promotion during identification of the $\bm \Xi$. One may also notice that $\|\cdot\|_0$ is not convex, which leads to a nonconvex optimization. One way to ameliorate this is to instead use the $\ell_1$ norm,resulting in the LASSO~\cite{Tibshirani1996}. The $\ell_1$ is convex but promotes a weaker form of sparsity. A computationally efficient alternative to directly solving the $\ell_0$-regularized problem was demonstrated in Brunton \textit{et al}.~\cite{brunton2016discover}, where a sequential threshold least-squares algorithm was used to greedily solve the problem; its convergence properties were proven in~\cite{zhang2019convergence}.

To encode the enforced energy-preserving constraint on the quadratics terms, Loiseau \textit{et al}.~\cite{loiseau_constrained_2018} vectorizes $\bm \Xi[:] = \bm \xi \in \mathbb R^{rN}$ to recast the constrained minimization problem as an unconstrained convex problem using an augmented functional formulation. If the constraint is encoded as $\bm C \bm \xi = \bm q$, together with the improvements above, the optimization problem reads,
\begin{equation}
    \begin{aligned}
    \bm \xi = \operatorname*{argmin}_{\bm \xi'}\max_{\bm \zeta}\biggr [& \frac{1}{2}\|\dot{\bm X}[:]-\bm{\hat \Theta}(\bm X)\bm \xi'\|^2_2 + \gamma\|\bm \xi'\|_1 \\
    &+ \bm \zeta^T(\bm C\bm \xi' - \bm q) \biggr ]
\end{aligned}
\end{equation}
where $\bm{\hat \Theta}(\bm X) = \bm I \otimes \bm \Theta(\bm X)$, $\dot{\bm X}[:]$ is the vectorized $\dot{\bm X}$ so that we have $\bm{\hat \Theta}\in \mathbb R^{rM\times rN}$, $\dot{\bm X}[:] \in \mathbb R^{rM}$. The number of constraints of the energy-preserving structure is given by $p = r(r+1)(r+2)/6$ and therefore we have $\bm C \in \mathbb R^{p\times rN}$, $\bm q\in \mathbb R^{p}$ and the Lagrange multiplier $\bm \zeta \in \mathbb R^p$. Additionally, for the linear equality constraints of the energy-preserving structure, $\bm q = 0$ and the number of free parameters is given by $rN - p = 2p$. If needed, these constraints can be used to impose any additional linear relationships from physical knowledge, such as enforcing global invariant preservation or enforcing \textit{a priori} known values of any entries $\xi_i$.

It is still not enough to guarantee the global stability of the linear-quadratic systems with the energy-preserving constraint. Based on the Schlegel and Noack trapping theorem, enforcing global stability for quadratically nonlinear dynamics requires searching for an $\bm m$ that makes $\tilde{\bm A}^S$ \textit{Hurwitz} without breaking the energy-preserving structure of the systems. In other words, the largest real part of the eigenvalue of $\tilde{\bm A}^S$ should be negative, so that the optimization problem becomes,
\begin{equation}
\begin{aligned}
    \label{opt:trappingSINDy}
    \bm \xi, \bm \zeta, \bm m = \operatorname*{argmin}_{\bm \xi', \bm m'}\ &\max_{\bm \zeta} \Biggr[ \frac{1}{2}\|\dot{\bm X}[:]-\bm{\hat \Theta}(\bm X)\bm \xi'\|^2_2 + \gamma\|\bm \xi'\|_1 \\
    & \left. + \bm \zeta^T(\bm C\bm \xi' - \bm q) + \frac{\bm \lambda_1(\tilde{\bm A}^S)}{\eta} \right],
\end{aligned}
\end{equation}
where the new term denotes that the largest eigenvalue of $\tilde{\bm A}^S$ will be enforced to be negative so that $\tilde{\bm A}^S$ is \textit{Hurwitz}, and the magnitude of whose loss is modulated by the hyperparameter $\eta$. One should notice that the last term of this problem is not convex but convex composite, making the problem non-convex, however the details of the algorithm used for solving the optimization problem in~\eqref{opt:trappingSINDy} will not be introduced in this work; see Kaptanoglu \textit{et al}.~\cite{kaptanoglu2021promoting} for more details.

\subsection{\label{subsec:extended_trapping} Proposed extended trapping SINDy algorithm}

Ideally, we would expect linear-quadratic systems with the quadratic energy-preserving structure to be globally stable and exhibit a monotonically trapping region according to Schlegel and Noack theorem. However, for quadratically nonlinear dynamics that do not have this energy-preserving structure, the stability is not guaranteed and generic quadratically nonlinear systems may be globally unbounded. Hence it is desirable to promote the local Lyapunov stability of the origin with some variations of the original trapping algorithm. For instance, problems such as open-channel flow and the Von Karman vortex street can weakly break the energy-preserving nonlinearities.

Similarly to Sawant \textit{et al}.~\cite{SAWANT2023115836} and Erichson \textit{et al}.~\cite{Erichson2019PhysicsinformedAF}, we can guarantee local stability and promote large stable radii for data-driven models with nonlinearities that do not preserve the energy by minimizing the destabilizing effects of the cubic terms in $\dot v(\bm y)$. The Sawant \textit{et al}.~\cite{SAWANT2023115836} optimization is approximately,
\begin{equation}
\begin{aligned}
    \label{opt:minimize Qijk}
    \bm \xi, \bm m = \operatorname*{argmin}_{\bm \xi', \bm m'}\biggr [& \frac{1}{2}\|\dot{\bm X}[:]-\bm{\hat \Theta}(\bm X)\bm \xi'\|^2_2 + \gamma\|\bm \xi'\|_1 \\
    &+ \alpha ^{-1}\|\bm H\|^2_F \biggr ],
\end{aligned}
\end{equation}
where the matrix $\bm H$ defined by~\eqref{eq:transform_T} can be extracted from the model coefficients. If the regularizer, here the $l_1$-norm, is convex then the optimization problem is convex. Sawant \textit{et al}. leverage this advantage and make the matrix $\tilde{\bm A}^S$ \textit{Hurwitz} via a post-processing step, making the optimization problem much easier to solve. However, this a-posteriori fixing of the \textit{Hurwitz} condition for $\tilde{\bm A}^S$ has the downside that it may take the solution significantly out of the local minimum found during optimization. For this reason, we combine the two loss functions in Eq.~\eqref{opt:trappingSINDy} and \eqref{opt:minimize Qijk} without the energy-preserving constraint to get,
\begin{equation}
\begin{aligned}
    \label{opt:minimize Qijk with trappingSINDy}
    \bm \xi, \bm m = \operatorname*{argmin}_{\bm \xi', \bm m'} \biggr [&\frac{1}{2}\|\dot{\bm X}[:]-\bm{\hat \Theta}(\bm X)\bm \xi'\|^2_2 + \gamma\|\bm \xi'\|_1 \\
    & + \eta^{-1}\bm \lambda_1(\tilde{\bm A}^S) + \alpha ^{-1}\|\bm H\|^2_F \biggr ].
\end{aligned}
\end{equation}
This loss function promotes a \textit{Hurwitz} $\tilde{\bm A}^S$ matrix while minimizing the destabilizing effects of the cubic contributions in $\dot v(\bm y)$. 
Moreover, $\tilde{\bm A}^S$ is defined by Eq.~\eqref{eq:def_AS_LS_dm} instead of Eq.~\eqref{eq:preservingAs} in this setting since the energy-preserving constraints are not prescribed. In principle, one could instead reformulate the last loss term as an inequality constraint, but we refrain from this strategy because it would involve inequality constraints that are not affine in the model coefficients. 

We now additionally promote locally stable models without minimizing the nonlinear effect of the systems. Thm.~\ref{thm:trapping_radius} enables local stability guarantees for quadratic systems with weakly energy-preserving nonlinearities. Recall that under certain conditions, linear-quadratic models can be locally bounded and exhibit a trapping region even if the hard constraint of the energy-preserving structure is relaxed. For ease of notation, we introduce $ Q'_{ijk} = Q_{ijk} + Q_{jki} + Q_{kij}$, $\bm H_0 = \mathcal{T}(\bm Q')$ and $\tilde{\bm H}_0 = \bm P^{-\frac{1}{2}}\mathcal{T}\bm H_0(\bm P^{-\frac{1}{2}}\otimes \bm P^{-\frac{1}{2}})$. The objective function now reads,
\begin{equation}
\begin{gathered}
\label{opt:trapping_inequality_constraint}
\bm \xi, \bm m = \operatorname*{argmin}_{\bm \xi', \bm m'}\left [ \frac{1}{2}\|\dot{\bm X}[:]-\bm{\hat \Theta}(\bm X)\bm \xi'\|^2_2 + \gamma\|\bm \xi'\|_1 + \frac{\bm \lambda_1(\tilde{\bm A}^S)}{\eta}\right ] \\
\text{s.t.} \quad \|\tilde{\bm H}_0\|_F = \epsilon_Q,
\end{gathered}
\end{equation}
or the form of unconstrained optimization,
\begin{equation}
\begin{aligned}
\label{opt:trapping_EP_loss_term}
\bm \xi, \bm m = \operatorname*{argmin}_{\bm \xi', \bm m'}\biggr[ &\frac{1}{2}\|\dot{\bm X}[:]-\bm{\hat \Theta}(\bm X)\bm \xi'\|^2_2 + \gamma\|\bm \xi'\|_1\\
&+ \eta^{-1}\bm \lambda_1(\tilde{\bm A}^S) + \beta^{-1}\|\tilde{\bm H}_0\|_F^2\biggr ].
\end{aligned}
\end{equation}
%

Eq.~\eqref{opt:trapping_EP_loss_term} is quite similar to the objective function used in Ouala \textit{et al}.~\cite{Ouala2023}. This optimization can be understood as maximizing the stability radius about the origin at $\bm m$ by promoting energy-preserving structure in $\bm Q$, instead of maximizing the stability radius about the radius by minimizing the norm of $\bm Q$ as in Sawant \textit{et al}.~\cite{SAWANT2023115836}, which makes the nonlinearities weak. When $\|\tilde{\bm H}_0\|_F^2$ is reduced to $0$, global stability is obtained and strict energy-preserving constraints can be found in $\bm Q$. In other words, the constrained optimization in~\cite{kaptanoglu2021promoting} is achieved using this unconstrained optimization. Moreover, by introducing the Lyapunov matrix $\bm P$, we generalize the energy-preserving structure which is restricted to the kinetic energy of the flow in the very original trapping SINDy algorithm. Considering that there is no ``true zero" due to the existence of machine precision in computer arithmetic; even the resulting models from Eq.~\eqref{opt:trappingSINDy} will not be truly globally bounded numerically. A description of each of the hyperparameters $\gamma, \eta$ and $\beta$ is provided in Table~\ref{tab:hyperparameters}.

Notice that we still need $\tilde{\bm A}^S$ to be \textit{Hurwitz} to get the estimated stability radius, and the largest eigenvalue $\lambda_1$ of $\tilde{\bm A}^S$ determines the stability radius directly. By additionally optimizing for the Lyapunov matrix $\bm P$, as in Goyal \textit{et al}.~\cite{goyal2023guaranteed}, we may be able to better control the magnitude of $\lambda_1$.

\begin{table*}[!t]
\caption{Description of extended trapping SINDy hyperparameters.}

\label{tab:hyperparameters}
\begin{tabular}{cp{0.97\textwidth}}
\hline\hline
$\gamma$ & Specifies the strength of sparsity-promotion through the $l_1$ regularizer. For simplicity, all results in this paper use $\gamma = 0$. \\ 
$\eta$ & Specifies how strongly to push the algorithm towards models with \textit{Hurwitz} $\tilde{\bm A}^S$.  \\
$\beta$ & Determines how strongly to penalize nonzero $\|\tilde{\bm H}_0\|_F$.\\
\hline \hline
\end{tabular}

\end{table*}

\subsection{\label{subsubsec:lsfunccomparison} Interpretation of the loss term}
Although the optimization settings of Eq.~\eqref{opt:minimize Qijk with trappingSINDy} and Eq.~\eqref{opt:trapping_EP_loss_term} seem to be similar and can both promote models that are locally stable by construction, the results are noticeably different. For Eq.~\eqref{opt:minimize Qijk with trappingSINDy}, this approach produces stable data-driven models with large stable radii by weakening the nonlinearities to get linearized models. However, it necessarily throws out  information regarding the nonlinear behavior of the physical systems, especially at large distances from the origin. In other words, this method of directly punishing the strength of the nonlinearities is primarily suitable for systems with weak nonlinearities.

Instead of minimizing the $\|\bm H\|_F $, Eq.~\eqref{opt:trapping_EP_loss_term} minimizes $\|\tilde{\bm H}_0\|_F$. To explain further, we do not promote the local stability of the models by weakening the nonlinearities of the system, but by enforcing a weakly quadratic energy-preserving structure. Moreover, we here use this optimization to find the $\rho_+$ and $\rho_-$, which guarantee the stability of the model based on Thm.~\ref{thm:trapping_radius}. Specifically, the size of the domain of boundedness $\rho_+$ and the size of the trapping region $\rho_-$ is derived from the optimized $\|\tilde{\bm H}_0\|_{max}$. While Ouala \textit{et al}.~\cite{Ouala2023} make a number of important contributions and extensions to neural-network-based system identification with stability, they are missing the $\rho_+$ and $\rho_-$ estimates derived in the present work that quantify the stability regions.


We can also reproduce the results from the work by Kaptanoglu \textit{et al}.~\cite{kaptanoglu2021promoting} through Eq.~\ref{opt:trapping_EP_loss_term}, which means our method could also promote global stability in quadratically nonlinear models. As what we show in Sec.~\ref{subsec:interp}, global stability of the models can be achieved, theoretically, by reducing the loss of $\|\tilde{\bm H}_0\|_F$ to machine precision.

\section{\label{sec:results}Results}
We now evaluate the performance and noise robustness of our extended trapping SINDy algorithm on tasks including model identification and simulation, for a few canonical systems which are commonly seen in fluid dynamics. Two widely used and appropriate definitions to assess the model quality are the time-average error $E_f$ in $\dot{\bm X}$, and for models with known closed forms, the relative Frobenius error $E_\text{coef}$ in $\bm \Xi$,
\begin{subequations}
\label{eq:error}
    \begin{align}
        E_f &= \frac{\overline{\|\dot{\bm X}_{\text{True}} - \dot{\bm X}_{\text{SINDy}}\|^2}}{\|\dot{\bm X}_{\text{True}}\|^2},\\
        E_\text{coef} &= \frac{\|\bm \Xi_{\text{True}} - \bm \Xi_{\text{SINDy}}\|_F}{\| \bm \Xi_{\text{True}} \|_F}.
    \end{align}
\end{subequations}
Another common quantity to evaluate the performance of data-driven models is the relative prediction error,
\begin{align}
    E_\text{pred} = \frac{\| \bm X_{\text{True}} - \bm X_{\text{SINDy}} \|_F}{\| \bm X_{\text{True}} \|_F}.
\end{align}
We also demonstrate the results qualitatively with dynamical trajectories of the model in order to examine, for instance, the noise robustness of the computational model. Another reason to exhibit results qualitatively rather than quantitatively is that quantifying the model performance for chaotic systems can be challenging because of sensitivity to initial conditions and aperiodic long-term behavior. However, despite the different ways we present for each example, we conduct each experiment in a similar approach. For each system, a single trajectory is used to train a SINDy model with the extended trapping optimizer. Keeping the same temporal duration, we then evaluate the model on different trajectories with new random initial conditions. Table~\ref{tab:results} summarizes the sampling, hyperparameters, and identified local stability estimates for a number of examples discussed in Sec.~\ref{sec:noisy_Lorenz}-\ref{sec:dysts_examples}. Several chaotic systems with different strange attractors will be presented in this section. Since it is usually difficult to quantify model performance for a chaotic system, we utilize Eq.~\eqref{eq:error} to investigate the utility of our extended trapping algorithm on these systems and compute the time-average error $E_f$ in $\dot{\bm X}$, rather than calculating the relative prediction error $E_\text{pred}$.

\begin{table*}[!t]
\caption{Description of the sampling, extended trapping SINDy hyperparameters, and identified local stability estimates for the dynamic systems examined in Sec.~\ref{sec:results}. $r$ denotes the dimension of the model. $\beta$ and $\eta$ control the minimization of the two terms in Eq.~\eqref{opt:trapping_EP_loss_term}. $\epsilon_Q$ reports the deviation from $\tilde{\bm H}_0$ having zero totally symmetric part. $\lambda_1$ denotes the largest eigenvalue in each $\tilde{\bm A}^S$ of the system.}
{\begin{tabular*}{\linewidth}{@{\extracolsep{\fill}} lccccccccccc }

\hline
\hline

Dynamic system    & $r$ & $\Delta t$ & $\beta$  & $\eta$ & $\epsilon_Q$ & $\rho_-$ & $\rho_+$ & $R_{eff}$ & $\lambda_1$ & $E_\text{coef}$ & $E_f$     \\ \hline
Lorenz Attractor  & 3 & 0.01     & $10^{-10}$ & $10^{5}$ & $8.2\times 10^{-7}$ & $103$    & $1.8\times 10^6$    & 4.2     & -0.98     & 0.013 & $10^{-5}$ \\
10\% noisy Lorenz & 3 & 0.01     & $10^{-14}$ & $10^{2}$ & $6.6\times 10^{-9}$ & $132$    & $2.2\times 10^8$    & 5.27     & -0.94     & 0.89 & 0.18    \\
50\% noisy Lorenz & 3 & 0.01     & $10^{-10}$ & $10^{4}$ & $0.01$              & $0$    & $0$    & $0$       & 0.05      & 1.3 & 0.9     \\
Finance       & 3 & 0.1 & $10^{-9}$ & $10^{2}$ & $8.2\times 10^{-10}$ & 32.6 & $1.8\times 10^8$ & 32.6 & -0.1 & 0.18 & $10^{-7}$ \\
Hadley        & 3 & 0.015 & $10^{-9}$ & $10^{2}$ & $2.8\times 10^{-10}$ & 21.3 & $5.3\times 10^8$ & 21.3 & -0.1 & 0.003 & $10^{-9}$  \\
LorenzStenflo & 4 & 0.039 & $10^{-9}$ & $10^{2}$ & $5.2\times 10^{-7}$ & 181 & $2.9\times 10^5$ & 181 & -0.1 & 0.03 & $10^{-7}$  \\
VallisElNino  & 4 & 0.022 & $10^{-9}$ & $10^{2}$ & $6.5\times 10^{-9}$ & $2.4\times 10^3$ & $2.3\times 10^7$ & $2.4\times 10^3$ & -0.1 & 0.7 & $10^{-5}$  \\ 
Von Karman (enstrophy) & $5$ & $0.1$ & $10^{-10}$ & $10^{4}$ & $5.9\times 10^{-10}$ & $4\times 10^{3}$ & $2.9\times 10^{8}$ & $2.4\times 10^{3}$ & $-0.12$ &  & \\ \hline \hline
\end{tabular*}}

\label{tab:results}
\end{table*}

The local stability size $\rho_+$ and the smallest trapping region size $\rho_-$ for each model are calculated if applicable for evaluation purpose. It is usually expected that we should observe an increase of $\rho_+$ and a decrease of $\rho_-$ with the number of optimization iterations. This trend is shown in Fig.~\ref{fig:iteration_illustration} and illustrates Thm.~\ref{thm:trapping_radius} and the two limits~\eqref{eq:limit}. To compare the smallest trapping region size across different systems we identify, a normalized quantity defined as $R_\text{eff} = \rho_-/\sqrt{\sum_{i=1}^{r}\bar{y}^2_i}$ is reported as well. 

\subsection{\label{sec:noisy_Lorenz} Noisy Lorenz attractor}

The well-known chaotic system developed by Lorenz in 1963~\cite{lorenz1963deterministic} has interested researchers for decades. General findings about the existence of a monotonically attracting trapping region are shown for quadratically nonlinear systems in Lorenz's work. Swinnerton-Dyer showed an attracting trapping region for the Lorenz system via the existence of a Lyapunov function outside the trapping region~\cite{SWINNERTONDYER2001}. The globally stable chaotic Lorenz system reads,
\begin{figure}
    \centering
    \includegraphics[width=\linewidth]{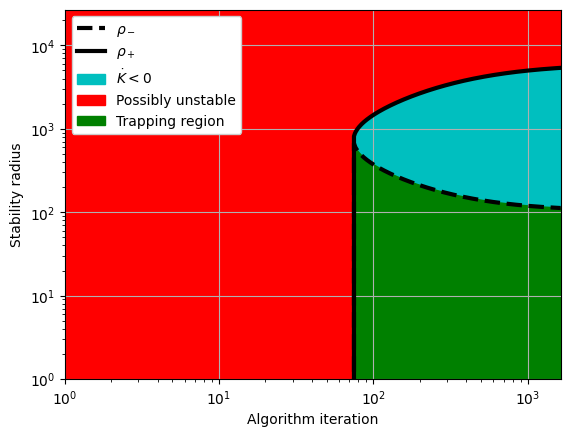}
    \caption{The progress of stability radius growth during optimization of identifying Lorenz attractor from data. The green area below the dashed line is the smallest ball-shape trapping region found by the model and the blue area between the solid line and the dashed line corresponds to the dotted area in Fig.~\ref{fig:trajectory_sketch}, where $\dot K < 0$. All trajectories in this area will eventually fall into the green one. The red area outside means no guarantees for the stability and trajectories may be long-time unstable.}
    \label{fig:iteration_illustration}
\end{figure}
\begin{subequations}
    \begin{align}
        \dot x &= \sigma(y-x),\\
        \dot y &= x(\rho - z) - y, \\
        \dot z &= xy - \beta z.
    \end{align}
\end{subequations}

\begin{figure}
\begin{center}
    \begin{subfigure}[b]{\linewidth}

    \includegraphics[width=\linewidth]{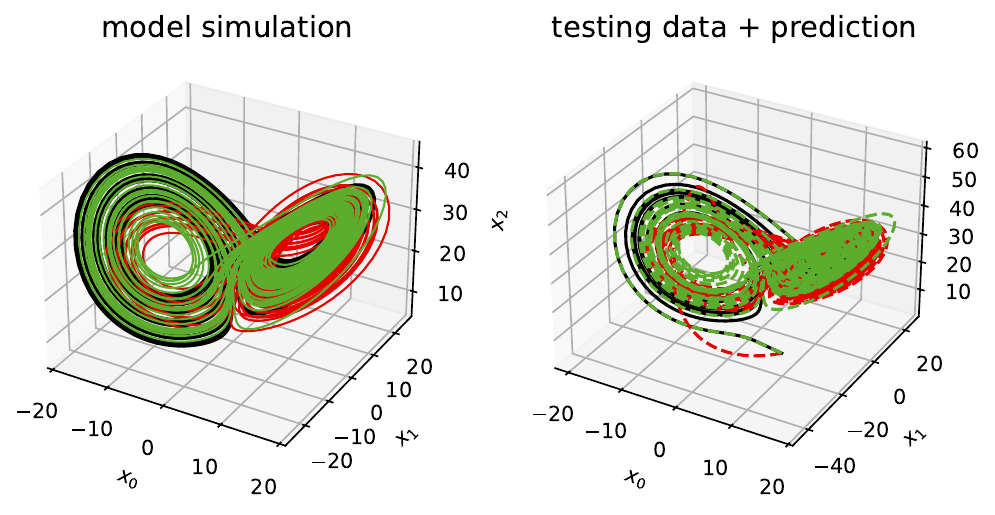}
    \caption{Lorenz system without noise.}
    \label{fig:Lorenz_no_noise}
    \end{subfigure}

    \begin{subfigure}[b]{\linewidth}

    \includegraphics[width=\linewidth]{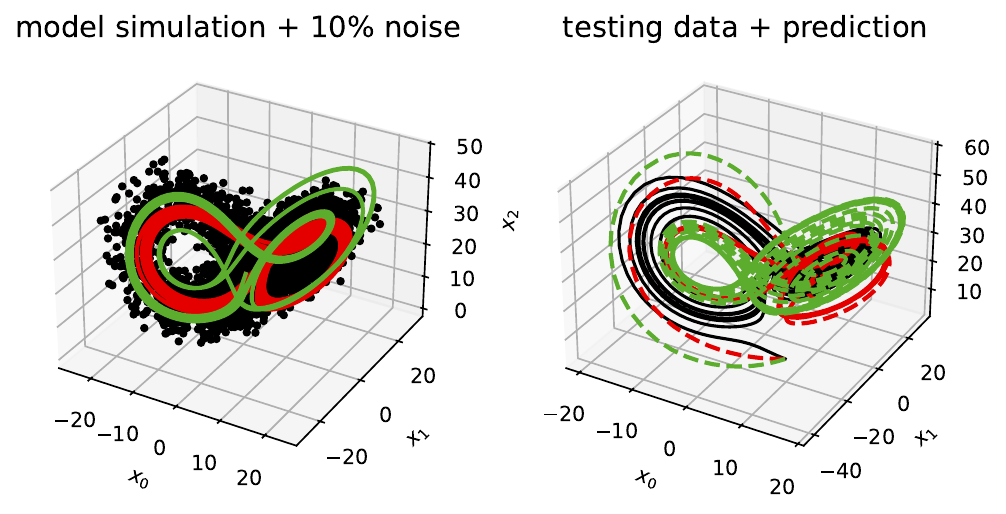}
    \caption{Lorenz system with 10\% noise.}
    \label{fig:comparison_noise_10}
    \end{subfigure}

    \begin{subfigure}[b]{\linewidth}

    \includegraphics[width=\linewidth]{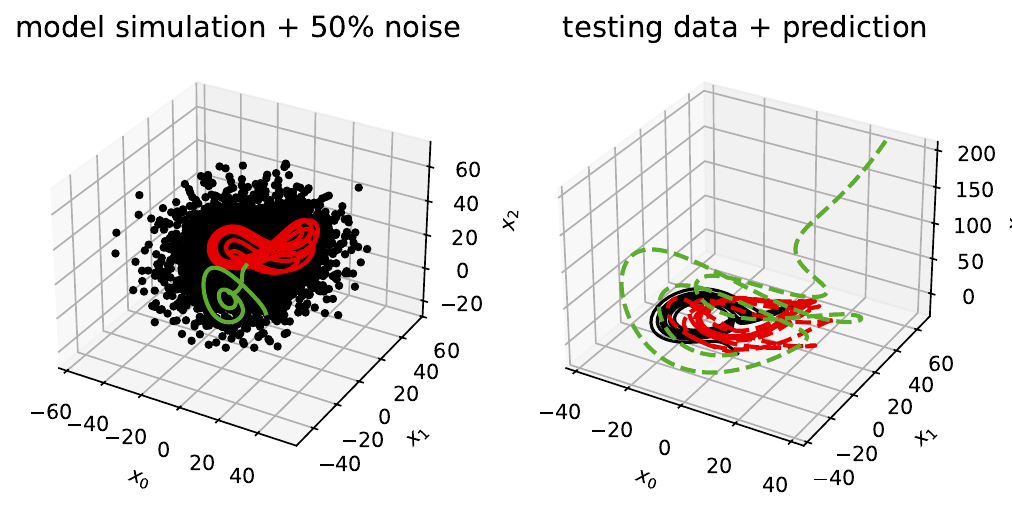}
    \caption{Lorenz system with 50\% noise.}
    \label{fig:comparison_noise_50}
    \end{subfigure}
\end{center}
    \caption[] 
        {A collection of Lorenz attractors with additive Gaussion noise. Left: Lorenz trajectories with noise for training (black), the corresponding extended trapping SINDy models (red) and the corresponding base SINDy models (green). Right: Lorenz trajectories without noise for test (black, their prediction by extended trapping SINDy models (red) and prediction by base SINDy models (green). } 
        \label{fig:noiselorenz}
\end{figure}
We choose the same parameters as in Lorenz 1963~\cite{lorenz1963deterministic}, which are $\sigma = 10$, $\rho = 28$, $\beta = 8/3$. We then add up to $50\%$ noise to the data with the choice of Lyapunov matrix $\bm P = \bm I$ to examine the noise robustness of the extended trapping algorithm. The noise is Gaussian noise with zero mean and the standard deviation given by the root mean square of the training trajectory. 

The results indicate that the algorithm can improve robustness to noise with respect to the identified model stability. To control for differential realizations of noise, we made twenty SINDy models from the same Lorenz training trajectory that receives added noise using realizations drawn from the noise distribution. We found that $17/20$ models were identified successfully  with local stability guarantees by the extended trapping algorithm in this work. At $50\%$ noise, this number drops to $10/20$ models. For comparison, the traditional SINDy algorithm produces $0/20$ provably stable models for both $25\%$ and $50\%$ noise added. Using SINDy with fairly tight energy-preserving constraints, but no enforcement of the Hurwitz property of $\tilde{\bm A}^S$, produces $10/20$ provably stable models at $25\%$ noise and $0/20$ at $50\%$ noise.

Fig.~\ref{fig:noiselorenz} illustrates the results with the extended trapping algorithm. With $10\%$ noise added, the attractor can be well identified. With $50\%$ noise added, in this case the identified model is still stable and represents a different strange attractor, but it fails to report a locally stable region. 
%
Another observation is that the $\tilde{\bm A}^S$ matrix of this $50\%$ noise model is not \textit{Hurwitz}, which has one and only one positive eigenvalue. This implies that even if the $\tilde{\bm A}^S$ is not negative-definite, the quadratically nonlinear system may still be locally stable. The reason could be that local stability holds and it is only our stability \textit{estimate} that has broken down. Or it may not be fully locally stable, but simply have only a very small or degenerate set of initial conditions that can in principle escape to infinity. We offer some potential insight into this scenario in Sec.~\ref{sec:conc}.

\subsection{Assorted locally stable models from the dysts database}\label{sec:dysts_examples}
The dysts dataset~\cite{gilpin2023chaos} is a collection of more than 100 chaotic systems ideal for benchmarking system identification methods~\cite{kaptanoglu2023benchmarking}. Four systems with the energy-preserving structure in the quadratic terms (and effective nonlinearity) are picked for benchmark testing as shown in Fig~\ref{fig:dysts_fig}. While these systems theoretically possess lossless nonlinearities, they cannot be maintained numerically. Nevertheless, the dataset remains valid for testing the extended algorithm and deriving stability estimates for the resulting numerical models. The datasets of all systems are normalized so we can compare the performance of one set of combinations of hyperparameters on different dynamical systems. The algorithm successfully recognizes the attractors and performs well on promoting stability in modeling these systems as well as the base SINDy, which is because we use clean data to train each model and the systems are all well-defined with large stability domain. However, base SINDy usually does not perform well on noisy data as we have shown in the previous example. Table~\ref{tab:results} also gives the normalized trapping region size of each system.

The finance chaotic system shown in Fig.~\ref{fig:finance} is defined in Cai \textit{et al}.~\cite{cai2007new}, and described as
\begin{subequations}
    \begin{align}
        \dot x &= z + \left (y-\alpha + \frac{1}{\beta} \right)x,\\
        \dot y &= 1 - \beta y - x^2, \\
        \dot z &= -x - \sigma z,
    \end{align}
\end{subequations}
where the parameters are $\alpha = 0.001$, $\beta = 0.2$ and $\sigma = 1.1$.
\begin{figure*}
        \begin{subfigure}[b]{0.252\textwidth}
            \includegraphics[width=\linewidth]{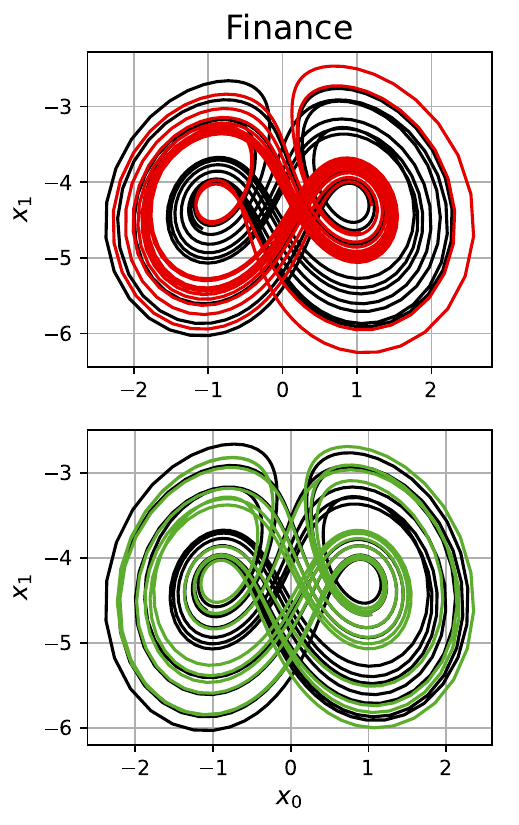}
            \caption[]%
            {{\small}}    
            \label{fig:finance}
        \end{subfigure}
        \begin{subfigure}[b]{0.23\textwidth}  
            \includegraphics[width=\linewidth]{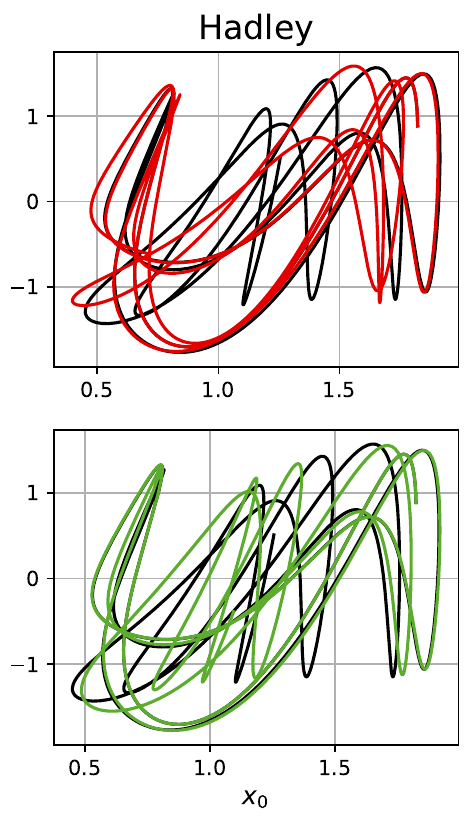}
            \caption[]%
            {{\small}}    
            \label{fig:hadley}
        \end{subfigure}
        \begin{subfigure}[b]{0.23\textwidth}   
            \includegraphics[width=\linewidth]{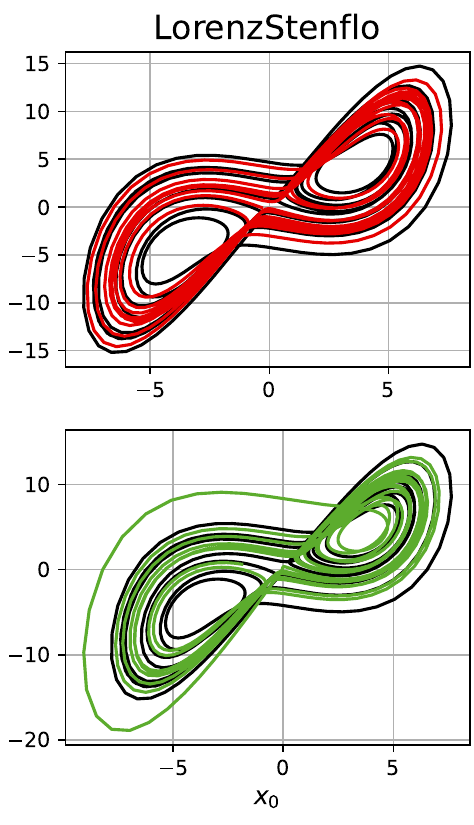}
            \caption[]
            {{\small}}
            \label{fig:lorenzstenflow}
        \end{subfigure}
        \begin{subfigure}[b]{0.245\textwidth}   
            \includegraphics[width=\linewidth]{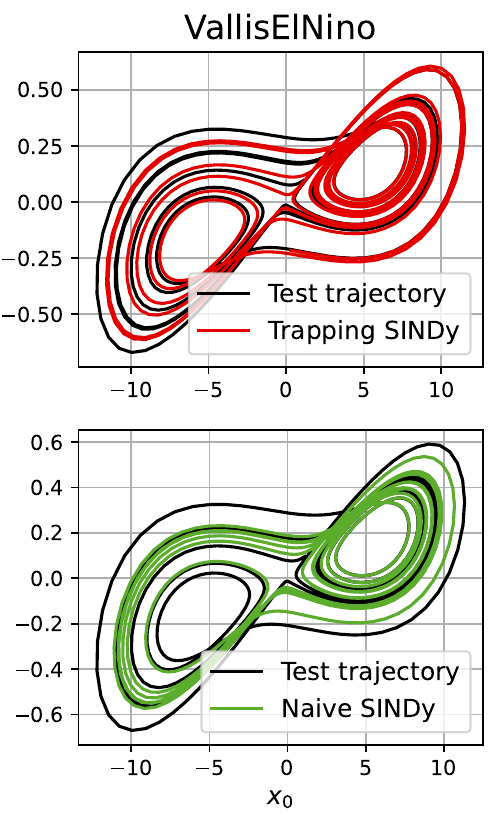}
            \caption[]%
            {{\small }}    
            \label{fig:valliselnino}
        \end{subfigure}
        \caption[]
        {A collection of benchmark test on dysts database. Trajectories of (a) a finance chaotic system, (b) Hadley cell dynamics, (c) Stochastic Lorenz–Stenflo system, (d) El Niño phenomenon for testing (black). Locally stable trapping SINDy models (red) capture and predict the strange attractors on these testing data, producing large stable radii for all four systems that satisfy the totally-symmetric quadratic coefficient constraints, each with a comparing prediction from accordingly base SINDy model (green).} 
        \label{fig:dysts_fig}
\end{figure*}
The Hadley system is an atmospheric convective cell which is also known as Hadley cell or Hadley circulation~\cite{hadley1735vi}. It was first envisioned in 1735 to explain the trade winds by Hadley. One way to describe the physical process of Hadley circulation is through this dynamical system, 
\begin{subequations}
    \begin{align}
        \dot x &= -y^2 - z^2 - ax + af,\\
        \dot y &= xy - bxz - y + g, \\
        \dot z &= bxy + xz - z,
    \end{align}
\end{subequations}
where the parameters are $a = 0.2$, $b = 4$, $f = 9$ and $g = 1$. The result of a 3D model for this test is shown in Fig.~\ref{fig:hadley}.

The Stochastic Lorenz–Stenflo system is a 4D system describing atmospheric acoustic-gravity waves by Stenflo~\cite{stenflo1996generalized}. The system reads,
\begin{subequations}
    \begin{align}
        \dot x &= \sigma(y-x) + sw,\\
        \dot y &= rx - xz - y, \\
        \dot z &= xy - bz,\\
        \dot w &= -x - \sigma w,
    \end{align}
\end{subequations}
where the parameters are $\sigma = 2$, $b = 0.7$, $r = 26$ and $s = 1.5$.

The El Niño phenomenon has a great impact on the global climate and a simple model to describe it is the Vallis continuous-time model~\cite{vallis1988conceptual}. The model is a 3D quadratically nonlinear system of ODEs and has three parameters:
\begin{subequations}
    \begin{align}
        \dot x &= By - C(x+p),\\
        \dot y &= -y + xz, \\
        \dot z &= -z - xy + 1,
    \end{align}
\end{subequations}
where $B = 102$, $C = 3$ and $p = 0$ results in the chaotic dynamics shown in Fig.~\ref{fig:valliselnino}. For the four systems, we straightforwardly identify models with large locally stable regions with the final $\rho_-$ and $\rho_+$ listed in Table~\ref{tab:results}.

\subsection{Von Karman vortex street using the enstrophy}\label{sec:von_karman_enstrophy}
The fluid wake behind a bluff body is characterized by a periodic vortex shedding phenomenon known as a von Kármán street. 
The two-dimensional incompressible flow past a cylinder is a typical example of such behavior, and has been a benchmark problem for Galerkin models for decades.
Early Galerkin models of vortex shedding, based on a POD expansion about the mean flow, captured the oscillatory behavior but were structurally unstable~\cite{Deane1991pof}.
This was later resolved by Noack \textit{et al}.~\cite{noack2003hierarchy}, 
in which
an 8-mode POD basis was augmented with a ninth ``shift mode'' parameterizing the mean flow deformation.
This approach was later formalized with a perturbation analysis of the flow at the threshold of bifurcation~\cite{Sipp2007jfm}.
\begin{figure*}
    \centering
    \includegraphics[width=\textwidth]{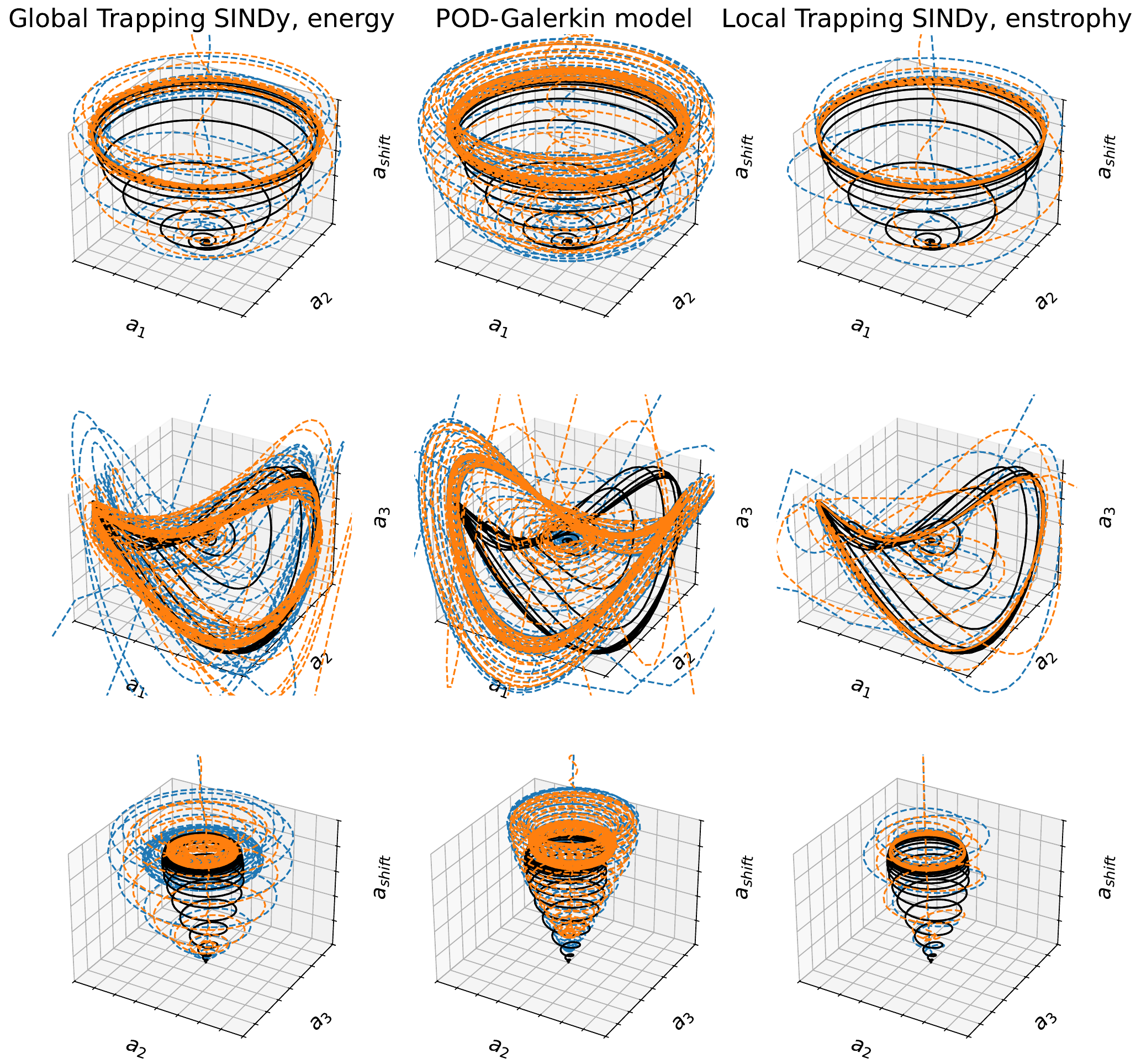}
    \caption{Five-dimensional models, integrated from three random initial conditions (color lines), for the von Karman street from the traditional trapping SINDy using the energy as a Lyapunov function, the analytic POD-Galerkin method, and the locally stable trapping SINDy using the enstrophy. The POD modes from the direct numerical simulation data were used for training and are shown in black lines. Each figure is plotted over the same axis dimensions and each row of the subplots shares the same combination of coordinates/modes.}
    \label{fig:von_karman_enstrophy}
\end{figure*}
%
The 9-mode quadratic Galerkin model from the Noack \textit{et al}. work 
accurately reproduces all of the key physical features of the vortex street. Moreover, in Schlegel and Noack~\cite{Schlegel2015} stability of the quadratic model was proven with $m_9 = m_\text{shift} = \epsilon$, $\epsilon > 1$, and $m_i = 0$ for $i = \{1,...,8\}$. \\
Recent work by Loiseau \textit{et al}.~\cite{loiseau_constrained_2018,loiseau2018sparse,loiseau2019pod} bypassed the Galerkin projection step by using the SINDy algorithm to directly identify the reduced-order dynamics.
This approach has been shown to yield compact, accurate models for low-dimensional systems ($r=2$ or $3$), but preserving accuracy and stability for higher-dimensional systems remains challenging. 
Higher-dimensional regression problems often become ill-conditioned; for instance, in the cylinder wake example, the higher modes 3-8 are essentially harmonics of the driving modes 1-2, and so it is difficult to distinguish between the various polynomials of these modes during regression. 
Because these higher harmonics are driven by modes 1-2, the 3D constrained quadratic SINDy model with modes 1-2 plus the shift mode from Loiseau \textit{et al}.~\cite{loiseau_constrained_2018} already performs well enough to capture the energy evolution with minor overshoot and correct long-time behavior.

In the original trapping SINDy work~\cite{kaptanoglu2021promoting}, new, provably globally stable models for the cylinder wake were found. In the present work, reusing the identical dataset described in that work, we extend these results by providing new locally-stable models, including one in which the enstrophy
\begin{equation}
    \bm P_\text{ens} = \begin{pmatrix}
       2.69 & 0.03 & 0.01 & 0 & -3.17 \\
         0.03 & 2.76 & 0 & 0 & 2.38\\
         0.01 & 0 & 10.64 & 0.08 & -0.06\\
         0 & 0 & 0.08 & 10.65 & -0.06\\
         -3.17 & 2.38 & -0.06 & -0.06 & 3
    \end{pmatrix}
\end{equation}
computed using the direct numerical simulation (DNS), is used as the Lyapunov matrix in Eq.~\eqref{eq:lyapunov_function}. Both trapping SINDy models show improved accuracy with respect to the POD-Galerkin model in recovering the true limit cycle dynamics in Fig.~\ref{fig:von_karman_enstrophy}.
\subsection{Lid-cavity flow}\label{sec:lid_cavity_flow}
\begin{figure*}
    \centering
    \includegraphics[width=\textwidth]{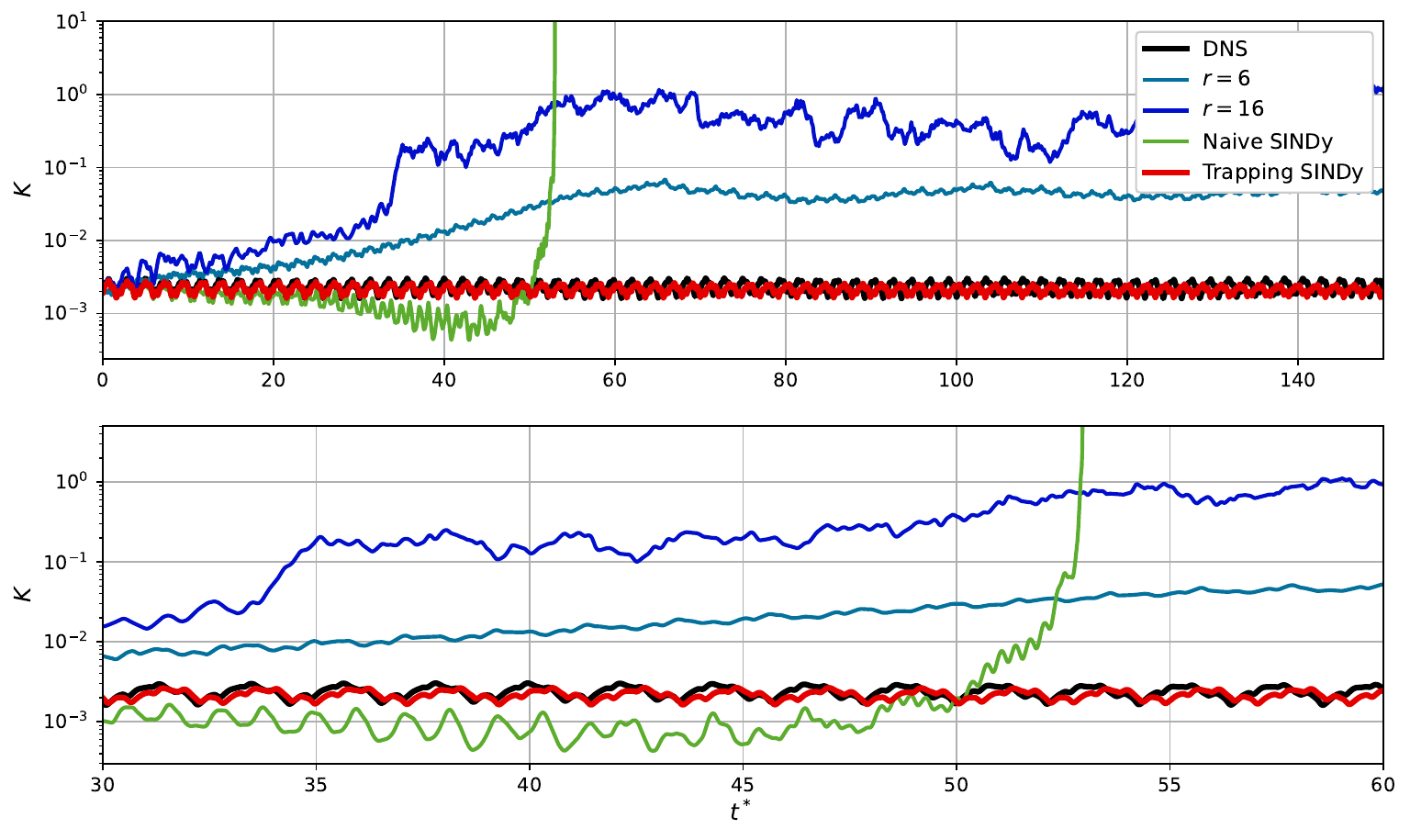}
    \caption{Lid-driven cavity, where $K$ denotes the total kinetic energy, and $t$ is the normalized time. Extended trapping SINDy and naive SINDy both use the first 6 POD modes. The lower subplots is a zoom-in of the upper subplots between time 30 to 60.}
    \label{fig:lid-cavity}
\end{figure*} 
In fluid mechanics, the lid-driven cavity is commonly used as a benchmark for testing numerical methods~\cite{kuhlmann2019lid}. A viscous fluid is in a confined geometry of a 2D rectangular container, and its motion is driven by the tangential moving of a bounding wall. The fluid motion is governed by incompressible Navier-Stokes equations. We then apply the model reduction method previously mentioned, resulting in a POD-Galerkin model. In this work, we select the first 6 POD modes to build a data-driven model.
The flow considered in the present work is the incompressible shear-driven cavity flow at $Re = 7500$. The computational domain and boundary conditions are the same as in Callaham \textit{et al}.~\cite{Callaham2022}. We also use the same data from DNS as in the work. Then the extended trapping SINDy algorithm is applied to compare with direct numerical simulation, 6-modes POD-Galerkin model, 16-modes POD-Galerkin model and a naive SINDy model. Since this system lacks lossless nonlinearities, applying the original trapping SINDy method for comparison would be inappropriate. As shown in Fig.~\ref{fig:lid-cavity}, the extended trapping SINDy model outperforms all other presented data-driven models and indicates strong stability compared to all others, even though the $\tilde{\bm A}^S$ from the extended trapping SINDy is not fully promoted to negative-definite. The naive SINDy model goes unstable immediately and POD-Galerkin models go unstable gradually, while the locally stable model captures the stable dynamics using the same or fewer number of modes. One may also notice the increasing phase shift between the DNS results and the trapping SINDy model, along with the reduction in total kinetic energy predicted by the trapping SINDy model. Both observations can likely be attributed to the same cause, which is that the trapping model takes only the first 6 POD modes to simulate the dynamics, and this truncation inevitably leads to a slight loss of energy. Additional results are provided in Appendix~\ref{appdx: power spectral}. The $\tilde{\bm A}^S$ has one positive eigenvalue which is similar as the one in the $50 \%$ noisy Lorenz example in Sec.~\ref{sec:noisy_Lorenz}. Further hyperparameter tuning might facilitate a fully provably stable model but we found empirically this model already works well for the lid cavity flow.

\section{\label{sec:conc}Conclusion}

Motivated by the nonlinear energy-preserving structure in multiple dynamical systems of fluids and plasmas, as well as the weakly nonlinear energy-preserving constraints in open-channel and many other important flows, the present work has considered the boundedness issue in quadratically nonlinear models and provided a new theorem with explicit bounds on the size of trapping regions for these models. An algorithm was proposed to identify such models from data based on this theorem.

In Sec.~\ref{sec:stability_analysis}, we analyzed the stability of linear-quadratic systems with weakly energy-preserving constraints on nonlinear terms. By adopting Lyapunov's direct method, a new theorem, which includes  sufficient conditions were proved for the existence of trapping regions in systems with such constraints and estimates were obtained of the size of these trapping regions. We also compared our theorem with Schlegel and Noack trapping theorem, which we showed that ours can exactly reproduce their trapping theorem.

In Sec.~\ref{sec:algorithm}, an extended trapping SINDy algorithm was proposed to promote local stability of data-driven quadratically nonlinear models. One of the guarantees on performance needed to use this algorithm is given by the local stability estimates from the theorem proposed in Sec.~\ref{sec:stability_analysis}. The algorithm performs well on tasks such as system identification and simulation with ROMs without requiring the nonlinearities of the dynamical system to be small to have these guarantees.
Successfully promoting the local stability of data-driven models requires a combination of well designed algorithms, methods of analysis with physical insight and interpretations, and local or global stability guarantees. The work of this paper was motivated, in part, by the idea that tools that extract models from data should guarantee the stability of models that capture the dominant behavior of bounded physical systems like fluids which can be described by Navier-Stokes equations. We observe many situations where the characteristics of the dynamics or stability/instability conditions are unclear in our presented example results and unpublished experiment, which includes that there exists one positive eigenvalue of $\tilde{\bm A}^S$ and that no stability guarantee is given by the Thm.~\ref{thm:trapping_radius} in some models but yet to find any unstable solution in neither of these situations. 

The robust stability despite weakly breaking the \textit{Hurwitz} property may be somewhat dynamically understood because when there is at least one positive eigenvalue the ellipsoid of positive energy growth, $\Omega_{\rho_-}$, transitions to a hyperboloid where the $\dot{v} > 0$ region stretches to infinity. Schlegel and Noack showed that generically this will be in conflict with global stability for any initial condition, although most trajectories in the $\dot{v} > 0$ region may remain bounded.

Indeed, generally the system dynamics can allow for trajectories starting in the $\dot{v} > 0$ region to move away from the origin, but then eventually cross the $\dot{v} = 0$ surface and move back towards the origin. Concluding definitively about local stability in this case is very challenging because of the strong dynamical dependence and this difficulty is in fact a motivation for the \textit{Hurwitz} assumption on $\tilde{\bm A}^S$ in order to write down a generic stability theorem. Empirically, the stability of these models appears relatively robust to a small positive eigenvalue in $\tilde{\bm A}^S$; it is only our stability \textit{estimates} that break down here. Nonetheless, if some eigenvalues of $\bm A^S$ are positive, and there are energy-preserving nonlinearities, it should be clear that $\dot{v} > 0$ is inevitable at distances far enough from the origin, implying instability for trajectories with initial condition in this area.

It roughly follows from this reasoning that if we promote models such that more and more eigenvalues of $\tilde{\bm A}^S$ become negative, this will contribute to the stability whether or not $\tilde{\bm A}^S$ is able to become fully \textit{Hurwitz}. This is precisely our empirical finding in both the local and global versions of the trapping SINDy algorithm. This leads to a number of possible directions for future work that would extend the research presented in this paper.

%
One direction in the future is to find the initial conditions that leads to an unstable solution that are still not clear in this work. Even though we have given sufficient conditions for attracting trapping regions to exist in quadratically nonlinear systems and conservative estimates about the size of the local stability domain, an unbounded trajectory has not been yet found when the \textit{Hurwitz} property is \textit{slightly} broken. The inability to find a $\bm m$ such that $\tilde{\bm A}^S$ is \textit{Hurwitz} can always be an issue of a suboptimal minima in the optimization problem, but dynamical explanations may be required too.


Future numerical work includes improving the strategy of enforcing a \textit{Hurwitz} matrix $\tilde{\bm A}^S$ during optimization. The current strategy to enforce a \textit{Hurwitz} matrix $\tilde{\bm A}^S$ is to solve $\bm \xi$ and $\bm m$ separately in each iteration so that each process is convex. However, the loss term for making a \textit{Hurwitz} $\tilde{\bm A}^S$ inevitably makes the loss function nonconvex. The convergence properties are unclear as well. A clever strategy for enforcing the \textit{Hurwitz} property was adopted recently in Goyal \textit{et al}.~\cite{goyal2023guaranteed}. Alternatively, with the help of modern powerful computational platforms, the nonconvex problems can be better solved by implementing this method in deep learning. So developing a deep neural network for this method as in Ouala \textit{et al}.~\cite{Ouala2023} is another promising line of future work.

Lastly, Kramer raises an optimization-based approach to enlarge the estimate of the stability region for a given Lyapunov matrix $\bm P$~\cite{kramer2021stability}. However, optimizing the Lyapunov matrix 
$\bm P$ could expand the stability radius $\rho_+$ as well. Thm.~\ref{thm:trapping_radius} suggests that tuning $\bm P$ to reshape $\tilde{\bm A}^S$’s eigenvalues may achieve this, as demonstrated in Sec.~\ref{subsec:extended_trapping}. Searching for an optimal $\bm P$ through data-driven methods as in Goyal \textit{et al}.~\cite{goyal2023guaranteed} and machine learning algorithms is promising to extend the algorithm in this work even further.

\begin{acknowledgments}
This work was supported by the National Science Foundation under Grant No. PHY-2108384 and National Science Foundation AI Institute in Dynamics Systems under Grant No. 2112085. 
\end{acknowledgments}

\section*{Data Availability Statement}
The data that support the findings of this work can be found in the open-source PySINDy package~\cite{deSilva2020,Kaptanoglu2022} where all examples studied in this work are implemented as well.

\section*{Author Declarations}
The authors have no conflicts to disclose.

\appendix
\section{Sensitivity to perturbations}\label{sec:sensitivity_analysis}
One very useful feature of our new local stability bounds in Theorem~\ref{thm:trapping_radius} is that we can provide estimates, for the case of quadratically nonlinear models, to a very interesting question in dynamical system identification from data; given the inevitable numerical and other errors incurred when identifying the coefficients of the right-hand-side of a dynamical equation from data, how large can the errors be in order to retain some large-scale (e.g. not just preserving the location of fixed points) sense of dynamical stability in the system?

Consider a machine-learned, quadratically nonlinear model $\dot x'_i$ from data, where, because of noise, nonoptimal optimization from hyperparameter or algorithm choices, or numerical approximation errors, there are small differences between the learned and true models 
\begin{align}
    \dot x_i - \dot x'_i = \delta C_i + \delta L_{ij}x_j + \delta Q_{ijk}x_jx_k.
    \end{align}
Now consider a system that is locally or globally bounded by exhibiting a monotonic trapping region as described in the present work. 
Without nonlinearity, or if the perturbation to the nonlinearity $\delta Q_{ijk}$ is somehow energy-preserving, then the requirement to retain local or global boundedness reduces to the simple requirement that (taking $\bm P = \bm I$ for clarity),
\begin{align}
(A^S)'_{jk}  &= L^S_{jk} - m_i Q_{ijk}  + \delta L^S_{jk} - m_i\delta Q_{ijk} \\ \notag &\equiv (A^S + \delta A^S)_{jk},
\end{align}
is \textit{Hurwitz}. Notice that $\bm A^S$ and the perturbation $\delta \bm A^S$ are necessarily real symmetric matrices that are diagonalized by orthogonal basis transformations. By assumption, there exists a $\bm m$ such that $\bm A^S$ is \textit{Hurwitz}, so that
\begin{align}
\bm y^T(\bm A^S)'\bm y &= \bm y^T\bm A^S\bm y + \bm y^T\delta \bm A^S\bm y \\ \notag &= \sum_{i=1}^r \lambda_i z_i^2 + \delta \lambda_i w_i^2 \leq \sum_{i=1}^r \lambda_1 z_i^2 + \delta \lambda_1 w_i^2,
\end{align}
where $\bm z = \bm V^T\bm y$ and $\bm w = \bm U^T\bm y$, with $\bm V$ and $\bm U$ denoting the orthogonal basis transformations. Therefore $\sum_i z_i^2 = \|\bm z\|_2^2 = \|\bm w\|_2^2 = \sum_iw_i^2$. Thus, 
\begin{align}
\bm y^T(\bm A^S)'\bm y \leq (\lambda_1 + \delta \lambda_1) \sum_{i=1}^r z_i^2 = (\lambda_1 + \delta \lambda_1) \|\bm y\|_2^2,
\end{align}
and since $\lambda_r < ... < \lambda_1 < 0$, it follows that a sufficient condition for being Hurwitz is that $\delta \lambda_1 > -\lambda_1$. This result is well known and merely restates that the perturbed $(\bm A^S)'$ matrix must remain \textit{Hurwitz}.

Next, we analyze the general situation where the perturbed nonlinear terms $\delta Q_{ijk}$ break the energy-preserving symmetry, as in the main body of the present work. It is a necessary condition that the perturbation $\bm A^S + \delta \bm A^S$ is still \textit{Hurwitz}, as described in the previous paragraph. Then the condition for the ``shell of stability'' to vanish in Theorem~\ref{thm:trapping_radius} is modified to, 
\begin{align} \notag
\epsilon_Q &= \|Q_{ijk} + Q_{jik} + Q_{kij} + \delta Q_{ijk} + \delta Q_{jik} + \delta Q_{kij}\|_F \\ &= \frac{3\lambda_1^2(\bm A^S + \delta \bm A^S)}{4\|\bm d + \delta \bm d\|_2}.
\end{align}
Note that the system may still have local stability -- it is our stability \textit{estimate} that breaks down. Nonetheless, this is a reasonable estimate for the identified system to break the dynamical stability of the true physical system.

This estimate allows us to roughly categorize quadratically nonlinear dynamical systems into three types: (1) systems very close to a stability threshold $\lambda_1 \approx 0$ that likely require explicit enforcement of the \textit{Hurwitz} constraint during optimization because of this sensitivity, (2) systems not close to the $\lambda_1 \approx 0$ stability threshold,  but close to the nonlinear threshold related to $\epsilon_Q$, where loss of local stability (estimates) may easily occur, and (3) systems which are robust to general perturbations in $\delta \bm A^S$ and $\delta \bm Q$. It turns out that many low-dimensional systems are relatively sensitive with regards to the local stability estimates in this work. For instance, an identified Lorenz model on noisy data might exhibit $\lambda \approx -1$, $\|\bm d + \delta \bm d\|_2 \approx \|\bm d\|_2 = 100$ and therefore a modest perturbation,
$$\epsilon_Q = \|\delta Q_{ijk} + \delta Q_{jik} + \delta Q_{kij}\|_F \approx \frac{3}{400} = 0.0075,$$
can break the local trapping region stability estimate. With respect to optimization errors or errors from noisy data, this is quite a small perturbation! Nonetheless, we stress again that, empirically, we often find robust local stability even where our stability \textit{estimates} are modestly broken. 

Moreover, this example shows that explicit estimates can be made of the sensitivity of a dynamical property (a local trapping region) to errors in nonlinear system identification, and provides evidence that generically the problem of correctly capturing the dynamical stability during data-driven model identification is \textit{ill-conditioned} for many dynamical systems, i.e. small perturbations in the coefficients make big changes in the dynamical stability considered here. Of course, we are always free to \textit{impose} local stability guarantees, whether or not the underlying data is compatible with the trapping theorems, since this can be seen as a form of model regularization. Then it becomes the role of the optimization to attempt to find a solution that both fits the data well and is compatible with the trapping theorems.
\begin{table*}[!t]
\centering
\caption{Example models discovered by the Extended Trapping SINDy algorithm compared to the true governing ODEs.}
{\begin{tabular*}{\linewidth}{ l|c|c }
\hline
\hline
Dynamic system  & True ODE &  Identified ODE  \\ \hline
Lorenz Attractor  & \multirow{3}{*}{\raisebox{-1.5\height}{\centering
    $\begin{aligned}
    \dot x &= 10(y-x)\\
    \dot y &= x(28 - z) -y\\
    \dot z &= xy - \frac{8}{3}z
\end{aligned}$%
}} & $\begin{aligned}
    \dot x &= 0.06 - 9.87 x + 9.94 y -0.01 z\\
\dot y &= -0.24 + 27.76 x -0.92 y + 0.03 z -0.99 x z\\
\dot z &= 0.08 -0.01 y -2.66 z + 0.99 x y
\end{aligned}$      \\\cline{1-1} \cline{3-3}
10\% noisy Lorenz & & $\begin{aligned}
    \dot x &= -0.01 + 3.67 x + 6.00 y -0.19 z -0.26 x z -0.04 y z\\
\dot y &= -0.01+ 5.55 x + 8.25 y + 0.22 z  -0.36 x z -0.23 y z\\
\dot z &= -0.06 -0.46 x + 0.32 y -2.63 z + 0.26 x^2 + 0.40 x y + 0.23 y^2 
\end{aligned} $      \\\cline{1-1} \cline{3-3}
50\% noisy Lorenz & & $\begin{aligned}
    \dot x &= -0.07 + 0.13 x + 3.57 y -1.23 z + 0.02 x y -0.02 x z -0.02 y^2 -0.07 y z + 0.03 z^2\\
\dot y &= 0.07 -0.37 x + 0.82 y + 0.95 z -0.02 x^2 + 0.02 xy -0.07 x z -0.02 y z -0.02 z^2\\
\dot z &= -0.01 + 1.38 x -1.41 y -0.41 z + 0.02 x^2 + 0.14 x y -0.03 x z + 0.02 y^2 + 0.02 y z
\end{aligned}$         \\ \hline
Finance           & $\begin{aligned}
    \dot x &= z + \left (y + 4.999 \right)x,\\
        \dot y &= 1 - 0.2 y - x^2, \\
        \dot z &= -x - 1.1 z,
\end{aligned}$ & $\begin{aligned}
    \dot x &=  4.93 x + 0.96 z + 0.99 x y - 0.01 y z \\
\dot y &= -0.01  - 0.20 y  - 0.99 x^2  + 0.01 x z \\
\dot z &= 0.01  - 1.02 x + 0.01 y  - 1.11 z\\
\end{aligned}$  \\ \hline
Hadley        & $\begin{aligned}
    \dot x &= -y^2 - z^2 - 0.2x + 1.8,\\
        \dot y &= xy - 4xz - y + 1, \\
        \dot z &= 4xy + xz - z,
\end{aligned}$ & $\begin{aligned}
    \dot x &= 1.80  - 0.20 x - 1.00 y^2  - 1.00 z^2\\
\dot y &= 1.00  - 1.00 y  - 0.01 z + 1.00 x y  - 3.99 x z \\
\dot z &=  0.01 y  - 1.00 z + 3.99 x y + 1.00 x z \\
\end{aligned}$   \\ \hline
LorenzStenflo & $\begin{aligned}
   \dot x &= 2(y-x) + 1.5w,\\
        \dot y &= 26x - xz - y, \\
        \dot z &= xy - 0.7z,\\
        \dot w &= -x - 2 w,
\end{aligned}$ & $\begin{aligned}
    \dot x &= 0.03  - 1.77 x + 1.95 y + 1.55 w  - 0.01 x z- 0.02 w^2\\
\dot y &= 0.51 + 25.41 x  - 0.90 y  - 0.05 z  - 0.46 w - 0.98 x z + 0.01 x w + 0.02 z w + 0.01 w^2\\
\dot z &= 0.01 + 0.01 x + 0.02 y  - 0.70 z + 0.01 x^2 + 0.98 x y  - 0.01 x w  - 0.02 y w  - 0.02 w^2\\
\dot w &= 0.03  - 1.06 x  - 0.01 y  - 0.01 z  - 2.37 w + 0.02 x w  - 0.01 y w + 0.02 z w\\
\end{aligned}$   \\ \hline
VallisElNino  & $\begin{aligned}
    \dot x &= 102y - 3x,\\
        \dot y &= -y + xz, \\
        \dot z &= -z - xy + 1,
\end{aligned}$ &  $\begin{aligned}
    \dot x &= -2.99 x + 101.64 y \\
    \dot y &= - 0.963 y + 0.01 z  + 0.99 x z + 0.03 y z  - 0.03 z^2\\
    \dot z &= 1.01  - 0.01 y  - 1.02 z  -0.99 x y - 0.03 y^2 + 0.03 y z\\
\end{aligned}$ \\ \hline \hline
\end{tabular*}}
\label{tab:add_results}
\end{table*}
\begin{figure*}[btp]
    \centering
    \includegraphics[width=\textwidth]{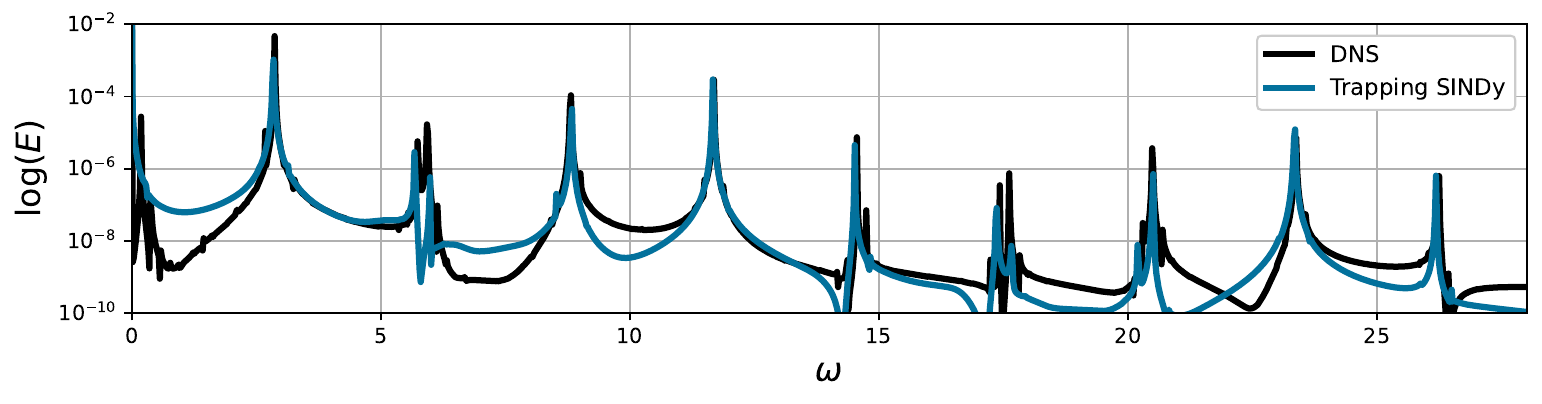}
    \caption{Power spectral density of the lid-driven cavity.}
    \label{fig:lid-cavity_psd}
\end{figure*} 

\section{Supplementary results}
\subsection{Models discovered by Extended Trapping SINDy}
We here show some models discovered by extended trapping SINDy algorithm in Table~\ref{tab:add_results}. Relevant terms are italicized for each equation. For readability, we present models discovered by the extended trapping SINDy algorithm with coefficients rounded to two decimal places (instead of the default omission of terms below 0.001).  One may notice that some models are generally non-sparse due to the fundamental challenge of handling noisy data.
\subsection{Power spectral density of the lid-driven cavity}
\label{appdx: power spectral}
We also check that if trapping SINDy model reproduces the power spectral density of the DNS data. Using FFT, we achieve basic power spectral density estimate as shown in Fig.\ref{fig:lid-cavity_psd} and the proposed method is able to capture the most significant frequencies of the lid-driven cavity. 

\onecolumngrid
\bibliography{sindy_stability}

\end{document}